\newcommand{\CLASSINPUTtoptextmargin}{2cm}
\newcommand{\CLASSINPUTbottomtextmargin}{2cm}

\documentclass[journal, twoside]{IEEEtran}
\IEEEoverridecommandlockouts                              % This command is only needed if 
\usepackage{graphics} % for pdf, bitmapped graphics files
\usepackage{epsfig} % for postscript graphics files
\usepackage{amsmath, amsthm} % assumes amsmath package installed
\usepackage{amssymb}  % assumes amsmath package installed
\usepackage{bm}
\usepackage{cite}
\usepackage{dsfont}
\usepackage{verbatim}
\usepackage{tikz}
\usetikzlibrary{shapes,arrows,positioning,calc}
\usepackage{epstopdf}
\usepackage{subfig}
\usepackage[font=footnotesize]{caption}
\usepackage{algorithm}
\usepackage{algorithmic}
\usepackage{graphicx}

\DeclareMathOperator{\rank}{rank}
\newcommand{\be}{\begin{equation}}
\newcommand{\ee}{\end{equation}}
\renewcommand{\IEEEQED}{\IEEEQEDopen}
\newtheorem{Remark}{\bf Remark}
\newtheorem{Lemma}{\bf Lemma}
\newtheorem{Corollary}{\bf Corollary}

\begin{document}
%\title{\LARGE \bf
%Diffusion Characterization and Control in Multi-Agent Networks
%}

%\author{\IEEEauthorblockN{
%	Wai Hong Ronald Chan\IEEEauthorrefmark{1},
%	Matthias Wildemeersch\IEEEauthorrefmark{2} and
%	Tony Q. S. Quek\IEEEauthorrefmark{3}\IEEEauthorrefmark{4}}
%\IEEEauthorblockA{
%	\IEEEauthorrefmark{1}Institute of High Performance Computing (IHPC), Singapore}
%\IEEEauthorblockA{
%	\IEEEauthorrefmark{2}International Institute for Applied Systems Analysis, Laxenburg, Austria}
%\IEEEauthorblockA{
%	\IEEEauthorrefmark{3}Singapore University of Technology and Design, Singapore}
%\IEEEauthorblockA{
%	\IEEEauthorrefmark{4}Institute for Infocomm Research (I$^2$R), Singapore}}
%\author{Wai Hong Ronald Chan$^{1}$, Matthias Wildemeersch$^{2}$ and Tony Q. S. Quek$^{3}$% <-this % stops a space
%\thanks{*This work was not supported by any organization}% <-this % stops a space
\title{Characterization and Control of Diffusion Processes\\in Multi-Agent Networks}
\author{Wai~Hong~Ronald~Chan, Matthias~Wildemeersch, Tony~Q.~S.~Quek%
\thanks{Wai Hong Ronald Chan is with Stanford University, California, USA (email: \tt{whrchan@stanford.edu}).}%
\thanks{Matthias Wildemeersch is with the International Institute for Applied Systems Analysis (IIASA), Laxenburg, Austria (email: \tt{wildemee@iiasa.ac.at}).}%
\thanks{Tony Q. S. Quek is with the Singapore University of Technology and Design (SUTD), Singapore (email: \tt{tonyquek@sutd.edu.sg}).}%
}
\markboth{}{Chan \textit{\MakeLowercase{et al.}}: Characterization and Control of Diffusion Processes in Multi-Agent Networks}
\maketitle
\thispagestyle{empty}

\tikzset{
block/.style = {draw, fill=white, rectangle, minimum height=3em, minimum width=3em},
tmp/.style  = {coordinate}, 
sum/.style= {draw, fill=white, circle, node distance=1cm},
input/.style = {coordinate},
output/.style= {coordinate},
pinstyle/.style = {pin edge={to-,thin,black}
}
}

\begin{abstract}

Diffusion processes are instrumental to describe the movement of a continuous quantity in a generic network of interacting agents. Here, we present a probabilistic framework for diffusion in networks and propose to classify agent interactions according to two protocols where the total network quantity is conserved or variable. For both protocols, our focus is on asymmetric interactions between agents involving directed graphs. Specifically, we define how the dynamics of conservative and non-conservative networks relate to the weighted in-degree Laplacian and the weighted out-degree Laplacian. Our framework allows the addition and subtraction of the considered quantity to and from a set of nodes. This enables the modeling of stubborn agents with time-invariant quantities, and the process of dynamic learning. We highlight several stability and convergence characteristics of our framework, and define the conditions under which asymptotic convergence is guaranteed when the network topology is variable. In addition, we indicate how our framework accommodates external network control and targeted network design. We show how network diffusion can be externally manipulated by applying time-varying input functions at individual nodes. %These functions can be constructed to accommodate various design objectives and network classes.
Desirable network structures can also be constructed by adjusting the dominant diffusion modes. To this purpose, we propose a Markov decision process that learns these network adjustments through a reinforcement learning algorithm, suitable for large networks. The presented network control and design schemes enable flow modifications that allow the alteration of the dynamic and stationary behavior of the network in conservative and non-conservative networks.

\end{abstract} 
\section{Introduction}
\label{sec:Introduction}

Large-scale network dynamics received ample research interest over the last decade in the context of group coordination~\cite{jadbabaie2003coordination}, distributed algorithms~\cite{boyd2006randomized,dimakis2010gossip}, network control \cite{kandhwaycampaigning}, distributed optimization~\cite{chen2013distributed}, consensus problems~\cite{acemoglu2011opinion,matei2012randomized}, or herding and flocking behavior~\cite{blondel2005convergence}. Network dynamics involve interactions between individual agents and relate to the diffusion of a continuous quantity within a generic network~\cite{ren1998diffusion,dargatz2007diffusion,dargatz2010bayesian,yildiz2010asymmetric,como2011robust,kiessling2011diffusion,acemoglu2013opinion}, including assets among financial institutions and opinions within social networks. In this work, we establish a probabilistic diffusion framework that describes in continuous time the movement of such a continuous quantity or node property within a multi-agent network. The main contributions of our framework are (i) to classify linear update rules and characterize the corresponding dynamical network behavior, (ii) to include external inputs as control variables into the framework resulting in the study of the network stability and convergence under different conditions, and (iii) to impose flow modifications by means of network control and network design. 

This framework builds on consensus models involving Markovian state transitions that have been formulated in discrete time~\cite{tcha1977optimal,jadbabaie2003coordination,olfati2004consensus,olfati2007consensus,acemoglu2008convergence,banisch2012agent} and continuous time~\cite{yildiz2011discrete}, as well as multi-agent gossiping models describing interactions between pairs of agents~\cite{acemoglu2013opinion}. We generalize these models and introduce a classification of linear inter-agent update rules depending on whether the global quantity initially present in the network is conserved or not. This enables our framework to account for a wider range of network phenomena. For example, financial or trade assets could be modeled using conservative flows, while opinions follow non-conservative network dynamics. Going beyond symmetric, unweighted graphs \cite{banerjee2008spectrum}, we focus on weighted graphs with asymmetric update rules and derive the corresponding differential equation that describes the diffusion of the considered quantity averaged over all sample paths. The nature of the update rule affects the stability and convergence of diffusion processes over networks \cite{blondel2005convergence}. We highlight the differences in the transient and stationary characteristics of conservative and non-conservative networks, the effects of network asymmetry, as well as the conditions for convergence in networks with switching topologies.

Building on existing leader-follower models~\cite{hong2006tracking,ji2008containment}, we extend the homogeneous differential equation that describes the diffusion process to its inhomogeneous form and use the inhomogeneous term to model an exogenous quantity input. By doing so, we can model the addition and subtraction of the considered quantity to and from the multi-agent network. Our framework unifies several types of systems involving external influences, including external restrictions on the quantity values of selected agents, as well as externally imposed reference values obtained from measurements or predefined system targets.  For instance, we demonstrate that diffusion processes considered by other authors, such as the presence of stubborn agents~\cite{acemoglu2013opinion} or the process of dynamic learning~\cite{acemoglu2008convergence}, can be expressed in terms of the inhomogeneous differential equation. In addition, we demonstrate that under some constraints on the input vector and the network topology, networks with exogenous excitations can remain stable and converge to a steady state. 

Our framework enables network control through the introduction of external excitations at individual nodes or the adjustment of the network structure. By transforming the basis of agents using standard results in state-space control to isolate the time dynamics of the characteristic modes of the system from one another~\cite{dahleh2004lectures}, and implementing a suitable controller for one of the modes, we can design time-varying input functions for each of the individual agents. A suitable choice of the controller will allow the network designer to guide the system towards a desired objective in terms of diffusion control. %The nature of the mode on which control is imposed also affects the resultant behavior of the system. 
Different types of networks are associated with different modal compositions and require different input design considerations for the achievement of control. Besides externally imposing control on a network, we also address the modification of inter-agent interactions to impose network control. By modifying the transition rate matrix that governs the network dynamics, the eigenvalues of the characteristic modes of the system are shifted, and the dynamics of the system are modified. Furthermore, we present an adaptive heuristic that models the modification of the network structure as a Markov decision process (MDP). In this MDP, the state space is defined as the set of individual agents, while the action space is defined as perturbations to the transition rate matrix governing the system evolution. Using an appropriate reward function and a suitable time schedule for decision making, the action space can be searched efficiently using a reinforcement learning algorithm based on standard learning techniques \cite{barto1998reinforcement}, and suboptimal control with quick convergence can be achieved. By enabling each agent to be associated with a cost or a reward, we allow the isolation and control of individual nodes through this modification. %This is above and beyond the optimal control processes formulated by other authors \cite{kandhwaycampaigning}, which engage in the control of groups of nodes based on their expected degree. 
In conclusion, the proposed network control and network design processes form a toolbox for the micromanagement of diffusion in networks.

%These control processes have varying levels of suitability for different types of networks. As an example, a network involving the diffusion of opinion may be more amenable to external control, while a network involving the diffusion of financial assets may be more amenable to either external control or internal network modifications depending on the availability of liquidity. Through either or both of these control processes, the dominant flow directions in diffusive networks can be modified so as to guide the diffusion of conservative flows, adjust the temporal characteristics of consensus achievement, or even modify the consensus objective. 

The remainder of the paper is structured as follows. In Section \ref{sec:Master}, we introduce two essential classes of stochastic update rules that are able to model diffusion in networks, and derive the corresponding linear system of ordinary differential equations that describe the network dynamics for the corresponding conservative and non-conservative networks. In Section \ref{sec:homoSC}, we further discuss the stability and convergence characteristics of these differential equations. In Section \ref{sec:inhomoSC}, we extend the homogeneous equations to their inhomogeneous forms, and discuss their stability and convergence characteristics. In Section \ref{sec:External}, we look at how external excitation can possibly achieve control, and in Section \ref{sec:Internal}, we discuss the impact and possible implementations of network structure modification. In Sections \ref{sec:External} and \ref{sec:Internal}, we also provide case studies to illustrate how these control processes can be implemented. In Section \ref{sec:Conclusion}, we conclude our work by looking at some future directions and extensions.
%In Section \ref{sec:Applications}, we discuss how our results can be applied to classes of networks, and illustrate some possible implementation strategies. 
\section{System model and stochastic update rules}
\label{sec:Master}

%Following the models in \cite{acemoglu2013opinion} and \cite{como2011robust}, 
We consider a population $\mathcal{V}$ of interacting agents $\mathcal{V}_i$, where $i \in \mathcal{I} = \{1,2,\ldots, n:n \in \mathbb{Z}^+\}$. These agents have the capacity to generate, destroy, store, transmit and receive a continuous quantity or node property $\mathcal{S}_i(t) \in \mathbb{R}$ that can change with time $t \geq 0$. We denote the vector of properties by $\mathcal{S}(t) = \{\mathcal{S}_i(t):i\in\mathcal{I}\}$. Given some set of initial conditions $\mathcal{S}_i(0) = \mathcal{S}_{i,0}$ for some $\mathcal{S}_{i,0} \in \mathbb{R}$, as well as probabilistic interactions between the agents, the system evolves over time and the property $\mathcal{S}_i(t)$ of each agent can be monitored over time. The probabilistic interactions between the agents can be described by %the values $[\mathcal{Q}_{ij}]$ or 
a weighted digraph $G = (\mathcal{V}, \vec{\mathcal{E}}, w)$ where $\mathcal{V}$ is the set of agents, $\vec{\mathcal{E}}$ is the set of directed links between pairs of agents in $\mathcal{V}$, and $w \colon \vec{\mathcal{E}} \to  \mathbb{R}^+$ is a function that represents the weights. %$\mathcal{Q}_{ij}$ is equivalent to the weight of the directed link between agent $j$ and agent $i$ in $\vec{\mathcal{E}}$ when $i \neq j$. 
The weighted adjacency matrix can be represented as 
\be
A_G(i,j) = 
\left\{
	\begin{array}{ll}
		w(i,j) & \mbox{if } (i,j) \in \vec{\mathcal{E}} \, , \\
		0 & \mbox{otherwise }\, .
	\end{array}
\right.
\ee
Furthermore, we define the weighted in-degree matrix $D_G^{(\mathrm{in})}$ and the weighted out-degree matrix $D_G^{(\mathrm{out})}$ as the diagonal matrices with
\be
D_G^{(\mathrm{in})}(j,j) = \sum_i A_G(i,j) \, ,
\ee
and
\be
D_G^{(\mathrm{out})}(i,i) = \sum_j A_G(i,j) \, .
\ee
The exact meaning of the edge weight and the direction of the links will be made explicit when we analyze the update rules, referred to as protocols, in the following Section. Since the interactions between agents can be asymmetric, we will introduce two different Laplacians that refer to the in-degree and the out-degree of each node. 
We define the weighted in-degree Laplacian as 
\be
L^{(\mathrm{in})}(G) = D_G^{(\mathrm{in})} -A_G\, .
\label{eqn:inLap}
\ee
Conversely, the weighted out-degree Laplacian is defined as 
\be
L^{(\mathrm{out})}(G) = D_G^{(\mathrm{out})} - A_G\, .
\label{eqn:outLap}
\ee
Depending on the choice of the stochastic update rule used in the probabilistic interactions between the agents, we will characterize the flow dynamics of networks operating under different protocols. Here, we describe two main classes of linear update rules that result in linear, time-invariant differential equations in the node property and corresponding linear, time-invariant matrix differential equations in the diffusion probabilities.\footnote{For non-linear update rules, we refer to \cite{SabMur:ACC03}.}  These update rules are distinguished by the conservation or the non-conservation of the total property initially present in the network. The networks applying the conservative and non-conservative protocols will be referred to as conservative and non-conservative networks respectively.

%%%%%%%%%%%%%%%%%%%%%%%%%%%%%%%%%%%%%%%%%%%%
\subsection{Conservative networks}
We first consider a protocol where the total amount of property in the network is conserved at every instant in time. Conservative updating is relevant for the description of conservative flow dynamics in a network, including the flow of material and physical assets. In this respect, conservative networks are able to represent stylized instances of hydraulic, financial, or trade networks. 

In conservative networks, agents obey the conservative update rule
\begin{align}
\mathcal{S}_i(t+\Delta t) &= \mathcal{S}_i(t) + \mathcal{C}_{ij} \mathcal{S}_j(t) \notag \\
\mathcal{S}_j(t+\Delta t) &= (1 - \mathcal{C}_{ij}) \mathcal{S}_j(t) \, , \tag{P1}
\label{eqn:P1}
\end{align}
where $i,j \in \mathcal{I}$, $(i,j) \in \vec{\mathcal{E}}$. The parameter $\mathcal{C}_{ij} \in (0,1]$ is a measure of liability or responsibility of agent $j$ towards agent $i$, and $\Delta t$ is an infinitesimal time interval. For every edge $(i,j) \in \vec{\mathcal{E}}$, there exists a clock that follows an independent Poisson process with rate $r_{ij} > 0$. The protocol \eqref{eqn:P1} is executed for nodes $i$ and $j$ when the independent Poisson clock of $(i,j)$ ticks. The following Lemma characterizes the property dynamics of the instance-averaged value of $\mathcal{S}$ in conservative networks.
\begin{Lemma} \label{lem:conPropDyn}
Let $\bar{\mathcal{S}}(t)$ denote the expected value of $\mathcal{S}$ averaged over all sample paths. The dynamics of the expected property for a network applying \eqref{eqn:P1} are defined by 
\be
\dot{\bar{\mathcal{S}}}(t) = \mathcal{Q}\bar{\mathcal{S}}(t) \, , \tag{A}
\label{eqn:markov2}
\ee
where $\mathcal{Q} = -L^{(\mathrm{in})}(G)$, the weighing function is defined as $w(i,j) = C_{ij} r_{ij}$, and 
\be
\mathcal{Q}_{ij} = 
\left\{
	\begin{array}{ll}
		C_{ij} r_{ij} & \quad \mbox{if } i \neq j \, , \\
		-\sum_{i\neq j} C_{ij} r_{ij} & \quad \mbox{if } i = j \, .
	\end{array}
\right.
\label{eqn:markovCoeffP1}
\ee 
\end{Lemma}
\begin{proof}
We first note that the total update rate for an node $i \in \mathcal{V}$ is given by $r_i = \sum_j r_{ij}$, and that the total update rate of the network is given by $r = \sum_i r_i$. Assume that a global network clock is ticking at rate $r$. Then, the probability that the clock will activate edge $(i,j)$ is given by $r_{ij}/r$, where in the limit of large-scale networks $r \approx 1/\Delta t$. Consequently, when property updating follows \eqref{eqn:P1}, the probabilistic update of the property of the nodes corresponding to the edge $(i,j)$ is given by
\begin{align}
\bar{S}_i(t+\Delta t) - \bar{S}_i(t) &= \sum_{j \neq i} r_{ij} \, \Delta t \, C_{ij}\,\bar{S}_j(t)\\
\bar{S}_j(t+\Delta t) - \bar{S}_j(t) &= - \sum_{i \neq j} r_{ij} \, \Delta t \, C_{ij}\,\bar{S}_j(t) \, ,
\end{align}
which can be succinctly written as
\be
\bar{\mathcal{S}}_i(t+\Delta t) - \bar{\mathcal{S}}_i(t) = \sum_{j \neq i}  \Delta t (r_{ij}C_{ij} \bar{\mathcal{S}}_j(t) - r_{ji}C_{ji} \bar{\mathcal{S}}_i(t)) \, .
\ee
Dividing by $\Delta t$ and taking the limit for $\Delta t \to 0$, we get a system of differential equations
\be
\dot{\bar{S}}_i(t) = \sum_{j \neq i} r_{ij}  C_{ij}\bar{S}_j(t) - r_{ji}C_{ji}\bar{S}_i(t) \, ,
\ee
or written differently
\be
\dot{\bar{S}}(t) = -L^{(\mathrm{in})}(G) \bar{S}(t) \, ,
\ee
which concludes the proof.
\end{proof}
We will refer to \eqref{eqn:markov2} as the governing equation. 
\begin{Corollary}
The matrix $\mathcal{Q} = -L^{(\mathrm{in})}(G)$ of a conservative network represents the transition rate matrix of a continuous-time Markov chain (CTMC). The transition probabilities $\mathcal{P}_{ij}(t)$ of the CTMC are solutions of the following differential equation
\be
\dot{\mathcal{P}}(t) = \mathcal{Q}\mathcal{P}(t) \, , \mathcal{P}(t) = [\mathcal{P}_{ij}(t)]\,.
\label{eqn:markov1Conservative}
\ee
\end{Corollary}
\begin{proof}
From \eqref{eqn:markovCoeffP1}, we notice that the 
\begin{equation}
\mathcal{Q}_{jj} \doteq -\sum_{i\neq j} \mathcal{Q}_{ij} \quad \forall \enskip j \in \mathcal{I},
\label{eqn:CTMC1Conservative}
\end{equation}
such that the columns of $\mathcal{Q}$ sum to zero. It follows that $\mathcal{Q}$ represents the transition rate matrix of a positive recurrent CTMC. Accordingly, the spread of property from one agent to another can be described by a CTMC whose state space is equivalent to the agent set $\mathcal{V}$ and whose instantaneous state is denoted by $X(t)$. Given that the diffusion dynamics can be described by a CTMC, we can write \cite{van2006performance}
\begin{equation}
\Pr[X(t+u)=i] = \sum_j \mathcal{P}_{ij}(u) \Pr[X(t)=j] \enskip, \forall \enskip u \in (0,\infty)
\label{eq:traPro}
\end{equation}
for some state $i \in \mathcal{I}$, where
\begin{equation}
\mathcal{P}_{ij}(u) = \Pr[X(t+u)=i|X(t)=j]
\end{equation}
describes the probability that an infinitesimal piece of property is in state $i$ (i.e. agent $i$) at time $t+u$, given that it was in state $j$ (i.e. agent $j$) at time $t$. From \eqref{eq:traPro}, it follows that the  transition probability matrix $\mathcal{P} = [\mathcal{P}_{ij}]$ satisfies 
\begin{equation}
\mathcal{P}(t+u) = \mathcal{P}(t)\mathcal{P}(u) \enskip\enskip \forall \enskip t,u \in (0,\infty) \, ,
\end{equation}
which is known as the Chapman-Kolmogorov equation. Assuming the limits $\lim_{h \rightarrow 0} \mathcal{P}(h) = \mathcal{P}(0)$ and $\lim_{h \rightarrow 0} \dot{\mathcal{P}}(h) = \dot{\mathcal{P}}(0)$ exist, we can then adopt the relation $\mathcal{P}(0) = \mathbb{I}$ where $\mathbb{I}$ is the identity matrix, and thereafter the following deterministic differential equation 
\begin{equation}
\dot{\mathcal{P}}(t) = \mathcal{Q}\mathcal{P}(t)
\end{equation}
where $\mathcal{Q} = \dot{\mathcal{P}}(0) = [\mathcal{Q}_{ij}]$.
\end{proof} 
We can interpret \eqref{eqn:markov1Conservative} as an equation describing the diffusion of transition probabilities between the various agents over time.

Note that Lemma \ref{lem:conPropDyn} can also be demonstrated via a probabilistic interpretation of the CTMC. It can be shown that $\mathcal{P}(t)$ is continuous in $t$, $\forall t \geq 0$ \cite{van2006performance}. For a conservative network, it follows from \eqref{eqn:CTMC1Conservative} that the columns of $\dot{\mathcal{P}}(t)$ always sum to zero. This shows that for each origin agent $j$, the sum over all the agents $i$ of the time-cumulative transition probabilities $\mathcal{P}_{ij}$  is conserved at a value of 1 over time, as expected from conventional laws of probability. We now consider the expected value of $\mathcal{S}_i$ averaged over all sample paths, which is given by the column vector
\begin{equation}
\bar{\mathcal{S}}(t) = \mathcal{P}(t)\bar{\mathcal{S}}(0) \, .
\label{eqn:instAvg}
\end{equation}
This illustrates that the sum of $\bar{\mathcal{S}}$ over all agents remains a conserved quantity when \eqref{eqn:P1} is applied. If the system starts from a known, deterministic state, then the bar in $\bar{\mathcal{S}}(0)$ can be omitted. Combining \eqref{eqn:markov1Conservative} and \eqref{eqn:instAvg}, we find again \eqref{eqn:markov2}.

%%%%%%%%%%%%%%%%%%%%%%%%%%%%%%%%%%%%%%%%%%%%
\subsection{Non-conservative networks}
We now consider a protocol where property diffuses between agents by means of a convex updating protocol. This protocol is of interest for opinion dynamics \cite{acemoglu2013opinion, dimakis2010gossip}, or preference dynamics in cultural theory \cite{thompson1990cultural}. 

In non-conservative networks, agents obey the following convex update rule \cite{acemoglu2013opinion}
\begin{equation}
\mathcal{S}_i(t + \Delta t) = \mathcal{C}_{ij}\mathcal{S}_j(t) + (1 - \mathcal{C}_{ij})\mathcal{S}_i(t) \, , \tag{P2}
\label{eqn:P2}
\end{equation}
whenever the Poisson clock ticks for a pair of agents $i,j \in \mathcal{I}$, $(i,j) \in \vec{\mathcal{E}}$. In other words, when the ($i,j$)-th Poisson clock activates the link between agents $i$ and $j$, agent $i$ is triggered to poll the property value of agent $j$ with a confidence measure of $\mathcal{C}_{ij}$ and update its own value accordingly. The following Lemma characterizes the property dynamics in non-conservative networks. 
\begin{Lemma} \label{lem:nonConPropDyn}
The dynamics of the expected property for a network applying \eqref{eqn:P2} are defined by the governing equation
\be
\dot{\bar{\mathcal{S}}}(t) = \mathcal{Q}\bar{\mathcal{S}}(t) \, , \tag{A}
\label{eqn:markov2}
\ee
where $\mathcal{Q} = -L^{(\mathrm{out})}(G)$, $w(i,j) = C_{ij} r_{ij}$, and 
\be
\mathcal{Q}_{ij} = 
\left\{
	\begin{array}{ll}
		C_{ij} r_{ij} & \quad \mbox{if } i \neq j \, , \\
		-\sum_{j\neq i} C_{ij} r_{ij} & \quad \mbox{if } i = j  \, .
	\end{array}
\right.
\label{eqn:markovCoeffP2}
\ee 
\end{Lemma}
\begin{proof}
When property updating follows \eqref{eqn:P2}, the probabilistic update of the property of node $v_i \in \mathcal{V}$ is given by
\begin{equation}
\bar{\mathcal{S}}_i(t+\Delta t) - \bar{\mathcal{S}}_i(t) = \sum_{j \neq i} (r_{ij} \Delta t) C_{ij} (\bar{\mathcal{S}}_j(t) - \bar{\mathcal{S}}_i(t)) \, .
\end{equation}
Dividing by $\Delta t$ and taking the limit for $\Delta t \rightarrow 0$, we obtain the instance-averaged linear differential equations represented by \eqref{eqn:markov2} with $\mathcal{Q} = -L^{(\mathrm{out})}(G)$. 
\end{proof}
Note that Lemma \ref{lem:nonConPropDyn} extends the basic consensus algorithm where $\dot{\bar{\mathcal{S}}}_i(t) = \sum_{j \in \mathcal{N}_i} \left(\bar{\mathcal{S}}_j(t) - \bar{\mathcal{S}}_i(t)\right)$, with $\mathcal{N}_i$ the neighborhood of $i$, to asymmetrically weighted updating. The relevance of the asymmetry will be further discussed in Section III. %should add cross-ref..  
\begin{Corollary}
The matrix $\mathcal{Q} = -L^{(\mathrm{out})}(G)$ of a non-conservative network represents the transition rate matrix of a CTMC, and the transition probabilities of the CTMC are solutions of the following differential equation
\begin{equation}
\dot{\mathcal{P}}(t) = \mathcal{Q} \mathcal{P}(t)\, ,
\label{eqn:markov1Consensus}
\end{equation}
\end{Corollary}
\begin{proof}
The proof is similar to the proof of Corollary 1, and follows from the fact that 
\begin{equation}
\mathcal{Q}_{ii} \doteq -\sum_{j\neq i} \mathcal{Q}_{ij} \quad \forall \enskip i \in \mathcal{I} \, ,
\label{eqn:CTMC1Consensus}
\end{equation}
meaning that the rows of $\mathcal{Q}$ all sum to zero. The CTMC corresponding to $\mathcal{Q}$ represents a random walk on $G$.
\end{proof}
In comparison with the transfer of an infinitesimal piece of property in a conservative network, $\mathcal{P}_{ij}$ in non-conservative networks represents the probability that node $i$ polls node $j$ to update its own value using the weight $\mathcal{C}_{ij}$. For networks applying \eqref{eqn:P2}, we note similarly that the rows of $\dot{\mathcal{P}}(t)$ always sum to zero, and consequently the rows of $\mathcal{P}(t)$ always sum to one. This shows that every polling tag originating from an agent $i$ must either still be under the possession of agent $i$ or of some other agent $j$ in the network. 

\begin{Remark}
Lemma \ref{lem:conPropDyn} and Lemma \ref{lem:nonConPropDyn} make the link between the protocols \eqref{eqn:P1} and \eqref{eqn:P2} and the transition rate matrices explicit. In the case of a symmetric $\mathcal{Q}$ matrix, the in-degree and out-degree Laplacians are identical. Hence, the conservative and non-conservative protocols are equivalent for symmetric matrices. In the event of asymmetry, we observe from \eqref{eqn:inLap} and \eqref{eqn:outLap} that the conservative $\mathcal{Q}$ and the transpose of the non-conservative $\mathcal{Q}$ are identical if the two corresponding digraphs have the same weights $w$ but with all link directions reversed.
\end{Remark}
\section{Homogeneous equation: Stability and convergence}
\label{sec:homoSC}

In this Section, we analyze the steady-state characteristics of the homogeneous equation \eqref{eqn:markov2} based on the eigendecomposition of the rate transition matrix $\mathcal{Q}$. The relevance of $\mathcal{Q}$ can further be understood by considering that  \eqref{eqn:markov1Conservative} and \eqref{eqn:markov1Consensus} can be solved as $\mathcal{P}(t) = \exp(\mathcal{Q} t)$, indicating the role of the topology, the inter-agent measures of confidence, and the updating rates in the diffusion process. Due to the construction of $\mathcal{Q}$ as a Laplacian matrix, $\mathcal{Q}$ always has the eigenvalue $q_\mathrm{s} = 0$. Moreover, $\mathcal{Q}$ is a Metzler matrix, such that $\exp(\mathcal{Q}t)$ is nonnegative for $t \geq 0$. If $\mathcal{Q}$ is diagonalizable, then we can write
\begin{align}
\bar{\mathcal{S}}(t) &= \exp \left(A \Lambda A^{-1} t \right) \mathcal{S}(0) \notag \\
&= A \, \text{diag} \left( \exp (q_k t) \right) A^{-1} \mathcal{S}(0) \, ,
\label{eqn:eigenMaster}
\end{align}
where $A$ contains the unit right eigenvectors of $\mathcal{Q}$ as columns, $A^{-1}$ contains the corresponding left eigenvectors of $\mathcal{Q}$ as rows, $q_k$ represents the eigenvalues of $\mathcal{Q}$, and $\Lambda = \text{diag}(q_k)$. We can further write
\begin{align}
\bar{\mathcal{S}}(t) &= \sum_k \exp(q_k t)v_{\mathrm{R},k}v_{\mathrm{L}^*,k}^T \mathcal{S}(0) \notag \\
&= \sum_k c_k \exp(q_k t) v_{\mathrm{R},k}  \, ,
\label{eqn:markov3}
\end{align}
where $v_{\mathrm{R},k}$ and $v_{\mathrm{L}^*,k}$ are the unit right and corresponding left eigenvectors\footnote{The left eigenvectors $v_{\mathrm{L}^*,k}$ are only unit in the case where $\mathcal{Q}$ is symmetric.} of $\mathcal{Q}$ in column form, and $c_k = v_{\mathrm{L}^*,k}^T \mathcal{S}(0)$ are scalars. We refer to the eigenvectors of $\mathcal{Q}$ as the system modes. From \eqref{eqn:markov3}, the eigenvalues and eigenvectors reflect the characteristic growth and decay rates of the system, as well as the dominant and weak nodes in the network. In particular, the quantities of the network nodes are stable and converge to a steady-state distribution in finite time iff the eigenvalues of $\mathcal{Q}$ are nonpositive and include zero. Note also that the steady-state behavior of the nodes is given by $\lim_{t \to \infty} \bar{\mathcal{S}}(t) = c_\mathrm{s} v_{\mathrm{R,s}}$, where $c_\mathrm{s}$ and $v_{\mathrm{R,s}}$ are the scalar and the unit right eigenvector corresponding to $q_\mathrm{s} = 0$. As per Ger\v{s}gorin's circle theorem, transition rate matrices $\mathcal{Q}$ constructed according to \eqref{eqn:P1} or \eqref{eqn:P2} and associated with CTMCs satisfy these properties. In the following subsections, we will illustrate the dynamics for different strongly connected network classes by means of the network presented in Fig. \ref{fig:verifyNetwork}.
\begin{figure}[!htb]
\centering
\begin{tikzpicture}[->,auto,node distance=1.6cm,>=latex',graph node/.style={circle,draw},every node/.style={scale=0.65}]
	
	\node [graph node] (1) {1};
	\node [graph node] (2) [right of=1] {2};
	\node [graph node] (3) [right of=2] {3};
	\node [graph node] (4) [right of=3] {4};
	\node [graph node] (5) [right of=4] {5};
	
	\path (1) edge [bend left] node {1} (2);
	\path (2) edge [bend left] node {1} (3);
	\path (3) edge [bend left] node {1} (4);
	\path (4) edge [bend left] node {1} (5);
	\path (5) edge [bend left] node [below] {$\alpha$} (4);
	\path (4) edge [bend left] node [below] {$\alpha$} (3);
	\path (3) edge [bend left] node [below] {$\alpha$} (2);
	\path (2) edge [bend left] node [below] {$\alpha$} (1);
	
\end{tikzpicture}
\caption{Path graph $P_5$ with symmetric links of unit weights used to illustrate the transient and steady-state behaviors of the various update rules.} \label{fig:verifyNetwork}
\end{figure}
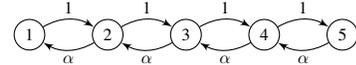
\vspace{-0.5 cm}
%%%%%%%%%%%%%%%%%%%%%%%%%%%%%%%%%%%%%%%%%%%%
\subsection{Conservative asymmetric networks}
For conservative networks, we formulate the following Lemma for the stationary behavior. 
\begin{Lemma} \label{lem:conP1}
The stationary value of a strongly connected asymmetric network applying \eqref{eqn:P1} is given by the unit steady-state right eigenvector $v_{\mathrm{R,s}}$ scaled by $c_\mathrm{s} = \sum_i S_{i,0}/\Psi$ where $\Psi$ is the sum of the entries of $v_{\mathrm{R,s}}$.
\end{Lemma}
\begin{proof}
In networks that apply \eqref{eqn:P1}, the column-sum of $\mathcal{Q} = - L^{(\mathrm{in})}(G)$ is zero, such that $\mathcal{Q}$ has a unit steady-state left eigenvector $v_{\mathrm{L,s}}$ with equal components. In an asymmetric network, the stationary node values are in general imbalanced due to the unequal components in the unit steady-state right eigenvector $v_{\mathrm{R,s}}$. Considering \eqref{eqn:markov3}, the proof is concluded.
\end{proof}
In the case of asymmetric liabilities between nodes, the stationary distribution will favor attractive nodes and shun repulsive nodes, and uniform spreading of the property over all the nodes does not occur in general. This is a consequence of the conservation of total property after every update. We illustrate this for the network in Fig. \ref{fig:verifyNetwork} with $\alpha = 0.2$ in the plot of $\bar{\mathcal{S}}_i(t)$ in Fig. \ref{fig:verifyAsymConservative}, where we observe that the steady-state properties for each node are different. In the latter plot, we generate $\bar{\mathcal{S}}_i(t)$  by both empirically generating sample paths based on the update rule \eqref{eqn:P1}, and analytically determining the expected property from \eqref{eqn:markov3}.
\begin{figure}[!t]
\centering
\includegraphics[width=0.8\linewidth]{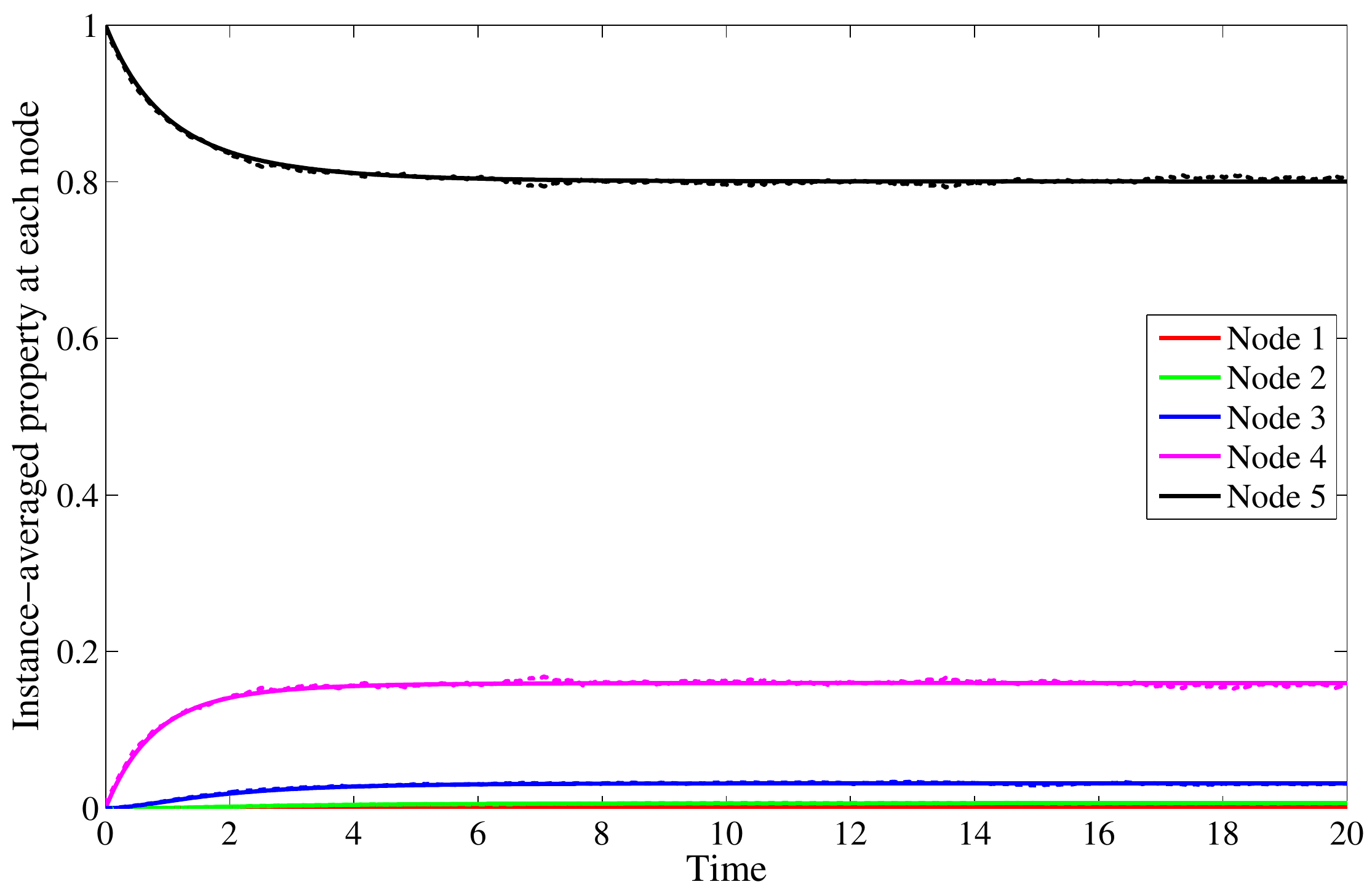}
\caption{Comparison of instance-averaged node property obtained through Monte Carlo simulations (dotted lines) and expected property obtained by eigendecomposition of the matrix $\mathcal{Q}$ (solid lines) for the asymmetric network in Fig. \ref{fig:verifyNetwork} with $\alpha = 0.2$ under \eqref{eqn:P1}. For the simulations, 5000 trials were performed, and each timestep was discretized into 1000 sub-units. In both cases, the initial conditions $\mathcal{S}(0) = [0\, 0\, 0\, 0\, 1]$ were adopted.}
\label{fig:verifyAsymConservative}
\end{figure}

%%%%%%%%%%%%%%%%%%%%%%%%%%%%%%%%%%%%%%%%%%%%
\subsection{Non-conservative asymmetric networks}
For networks that apply the convex update rule \eqref{eqn:P2}, we formulate the following Lemma for the stationary behavior. 
\begin{Lemma} \label{lem:conP2}
A strongly connected asymmetric network following the convex update rule \eqref{eqn:P2} always achieves consensus. The final consensus value $c_\mathrm{v} = c_\mathrm{s} / \sqrt{n}$ in the infinite time horizon is the weighted average of $\mathcal{S}(0)$, where the weights are the corresponding entries of the unit steady-state left eigenvector $v_{\mathrm{L,s}}$, scaled by the factor $1/\Omega$, where $\Omega$ is the sum of the entries in $v_{\mathrm{L,s}}$. In other words, 
\begin{equation}
c_\mathrm{v} = \dfrac{1}{\Omega} v_{\mathrm{L,s}}^T\mathcal{S}(0)\, .
\end{equation} 
\end{Lemma}
\begin{proof}
Since the row-sum of $\mathcal{Q} = - L^{(\mathrm{out})}(G)$ is zero, $\mathcal{Q}$ has a right unit eigenvector with equal components, demonstrating that strongly connected networks with possibly asymmetric updating always achieve agreement. In order to determine the final consensus value, we need to determine for each node the constant $c_\mathrm{s}$ that corresponds to the initial conditions enforced on the system. This involves finding the row of $A^{-1}$ corresponding with $q_\mathrm{s} = 0$. For an asymmetric matrix, this corresponds to the unit steady-state left eigenvector of $\mathcal{Q}$ scaled by the factor $\sqrt{n}/\Omega$, since the dot product of this vector and the unit steady-state right eigenvector must equal 1 based on the initial conditions $\mathcal{P}(0) = \mathbb{I} \, = A^{-1}A$. The result is then scaled by the magnitude of the entries of $v_{\mathrm{R,s}}$ to obtain $c_\mathrm{v}$.
\end{proof}

We illustrate this for the network in Fig. \ref{fig:verifyNetwork} with $\alpha = 0.2$ in the plot of $\bar{\mathcal{S}}_i(t)$ in Fig. \ref{fig:verifyAsymConsensus}. In the latter plot, we generate $\bar{\mathcal{S}}_i(t)$  by both empirically generating sample paths based on the update rule \eqref{eqn:P2}, and analytically determining the expected property from \eqref{eqn:markov3}. In this case, the final consensus value is heavily weighted towards node 5, whose inward link weight exceeds its outward link weight. This means that node 5 is heavily polled by its neighbors, resulting in its higher weight in the final consensus value. Such a trend is reminiscent of the availability heuristic, or availability bias, where subjects that are encountered or recalled more often are given more thought and emphasis~\cite{tversky1973availability}.

\begin{figure}[!t]
\centering
\includegraphics[width=0.8\linewidth]{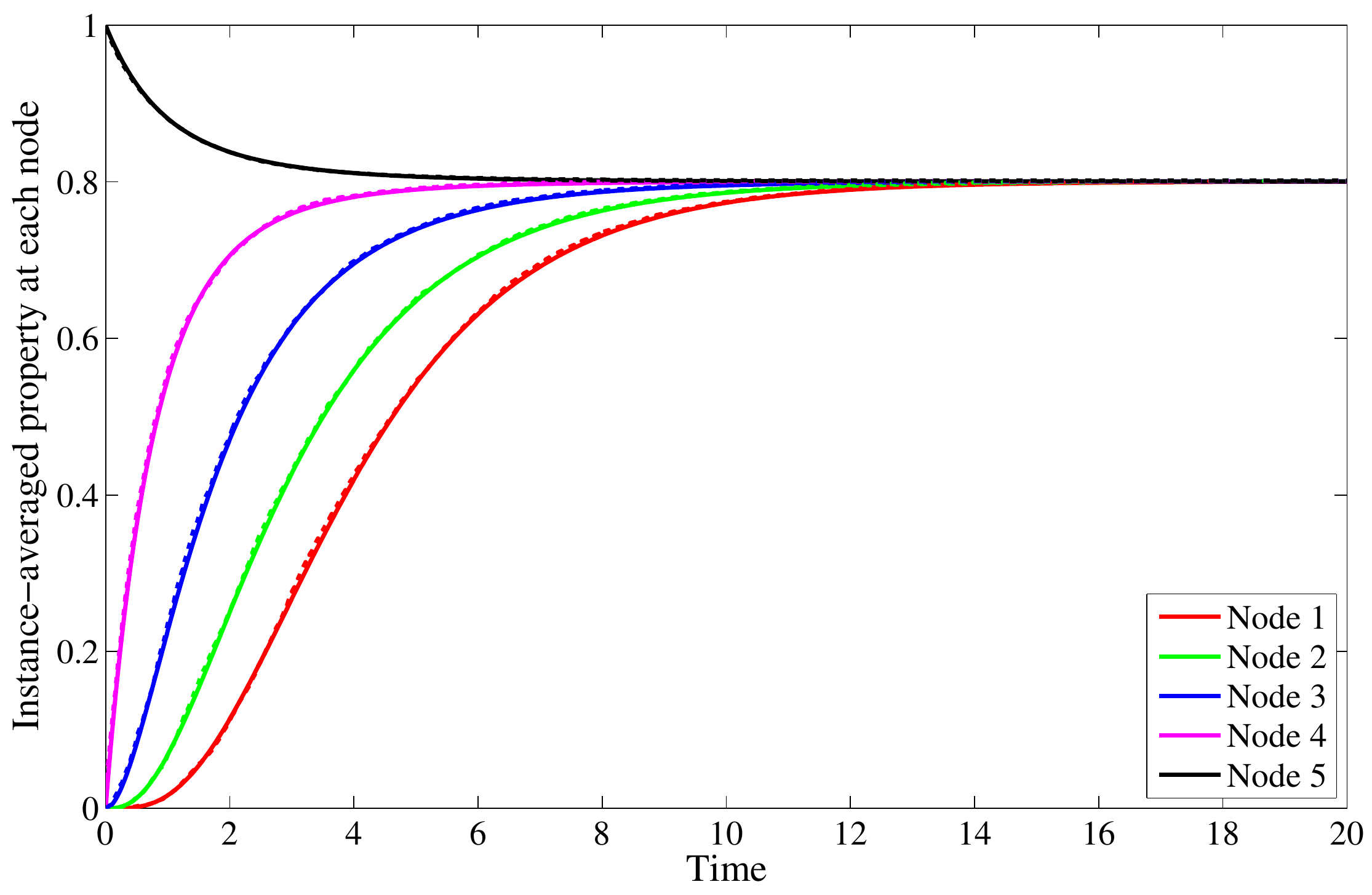}
\caption{Comparison of instance-averaged property obtained numerically (dotted lines) and analytically (solid lines) for the asymmetric network in Fig. \ref{fig:verifyNetwork} with $\alpha = 0.2$ under \eqref{eqn:P2}. The simulation parameters are identical to those in Fig. \ref{fig:verifyAsymConservative}. In this case, the unit steady-state left eigenvector is $[0.0016 \, 0.0078 \, 0.0392 \, 0.1960 \, 0.9798]$ and the unit steady-state right eigenvector is $\tfrac{1}{\sqrt{5}}[1 \, 1 \, 1 \, 1 \, 1]$, which gives us $\Omega = 1.2244$ and the corresponding consensus value $c_\mathrm{v} = 0.8003$.}
\label{fig:verifyAsymConsensus}
\end{figure}

\subsection{Symmetric networks}
In a strongly connected symmetric network, the steady-state right eigenvector has components of equal magnitude in all dimensions regardless of the update rule. We can interpret this as an equal sharing of resources in conservative networks and consensus among all agents at a value corresponding to the average of the initial conditions at all the nodes in non-conservative networks \cite{olfati2004consensus,olfati2007consensus,acemoglu2008convergence,acemoglu2013opinion}. We illustrate this for the network in Fig. \ref{fig:verifyNetwork} with $\alpha = 1$ in the plot of $\bar{\mathcal{S}}_i(t)$ in Fig. \ref{fig:verifySym}. In the latter plot, we generate $\bar{\mathcal{S}}_i(t)$  by both empirically generating sample paths based on the update rules \eqref{eqn:P1} and \eqref{eqn:P2}, which are equivalent in this case, and analytically determining the expected property from \eqref{eqn:markov3}.

\begin{figure}[!t]
\centering
\includegraphics[width=0.8\linewidth]{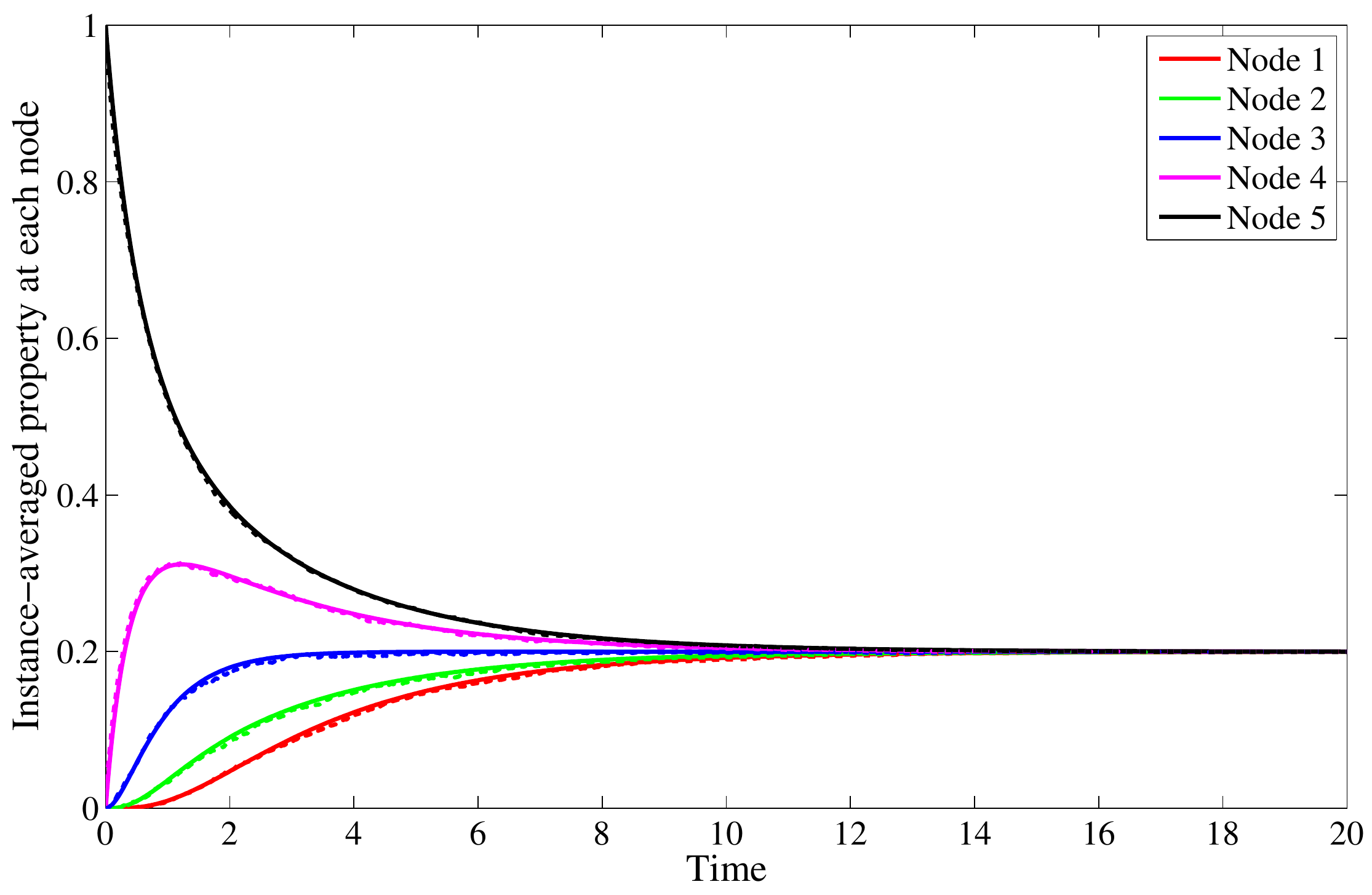}
\caption{Comparison of instance-averaged property obtained numerically (dotted lines) and analytically (solid lines) corresponding to the symmetric network in Fig. \ref{fig:verifyNetwork} with $\alpha = 1$. The simulation parameters are identical to those in Fig. \ref{fig:verifyAsymConservative}.}
\label{fig:verifySym}
\end{figure}

%%%%%%%%%%%%%%%%%%%%%%%%%%%%%%%%%%%%%%%%%%%%
\subsection{Switching topologies}
There are many important scenarios where the network topology can change over time. As an example, connections in wireless sensor networks can be established and broken due to mobility. In this Section, we study the steady-state behavior of networks with switching topologies and provide a condition that leads to convergence under dynamic topologies. Networks with switching topologies converge if they exhibit the same invariant stationary value as defined in Lemma \ref{lem:conP1} and Lemma \ref{lem:conP2}, which depend on $v_{\mathrm{R},s}$ and $v_{\mathrm{L},s}$ for \eqref{eqn:P1} and \eqref{eqn:P2} respectively. 

We introduce now two classes of strongly connected, directed graphs  for which the weighted Laplacian shares the left or the right eigenvector corresponding to $q_\mathrm{s} = 0$. For conservative networks, we have
\begin{multline}
\mathcal{G}_{v_\mathrm{s}}^\eqref{eqn:P1} = \lbrace  G(\mathcal{V}, \vec{\mathcal{E}}, w)\, : \, \rank L^{(\mathrm{in})}(G) = n-1, \\ 
L^{(\mathrm{in})}(G) \times v_\mathrm{R,s}  = \vec{0} \rbrace \, .
\end{multline}
For non-conservative networks, we define
\begin{multline}
\mathcal{G}_{v_\mathrm{s}}^\eqref{eqn:P2} = \lbrace  G(\mathcal{V}, \vec{\mathcal{E}}, w)\, : \, \rank L^{(\mathrm{out})}(G) = n-1, \\
 v_\mathrm{L,s}  \times L^{(\mathrm{out})}(G) = \vec{0} \rbrace \, .
\end{multline}
Correspondingly, we define the following sets of matrices with equal steady-state eigenvectors as follows
\begin{align}
M_{v_\mathrm{s}}^\eqref{eqn:P1}  & =  \lbrace \mathcal{Q}\colon \mathcal{Q} = - L^{(\mathrm{in})}(G) , \mathcal{Q} \times v_\mathrm{R,s} = \vec{0} \rbrace \nonumber \\
 M_{v_\mathrm{s}}^\eqref{eqn:P2}  & =  \lbrace \mathcal{Q}\colon \mathcal{Q} = - L^{(\mathrm{out})}(G) , v_\mathrm{L,s} \times \mathcal{Q} = \vec{0} \rbrace \, .
\end{align}
For these classes of networks, the system dynamics are described by
\be \label{eqn:swiTop}
\dot{\bar{\mathcal{S}}}(t) = \mathcal{Q}_k \bar{\mathcal{S}}(t) \, ,
\ee
where $\mathcal{Q}_k \in M_{v_\mathrm{s}}^\eqref{eqn:P1}$ when \eqref{eqn:P1} is applied, and $\mathcal{Q}_k \in M_{v_\mathrm{s}}^\eqref{eqn:P2}$ when \eqref{eqn:P2} is applied. The index $k$ changes over time according to $k = f(t), \, f\colon \mathbb{R}^+ \to \mathcal{I}_M$, where $\mathcal{I}_M$ is the index set of the corresponding set of matrices. 

\begin{Lemma}
Consider a network with network dynamics as described in \eqref{eqn:swiTop}. The networks that are strongly connected and follow \eqref{eqn:P1} or \eqref{eqn:P2} are globally asymptotically convergent if $\mathcal{Q}_k$ belongs to $M_{v_s}^\eqref{eqn:P1}$ or $M_{v_s}^\eqref{eqn:P2}$, respectively. 
\end{Lemma}
\begin{proof}
The proof follows from the fact that all networks that belong to $M_{v_\mathrm{s}}^\eqref{eqn:P1}$ or $M_{v_\mathrm{s}}^\eqref{eqn:P2}$ converge to the same invariant quantity, as they share the same right and left eigenvector, respectively. As all elements of $M_{v_\mathrm{s}}^\eqref{eqn:P1}$ and $M_{v_\mathrm{s}}^\eqref{eqn:P2}$ are globally asymptotically stable, this concludes the proof. 
\end{proof}

We introduce now the distance vector $\delta = \bar{\mathcal{S}}(t) - c_\mathrm{s} v_{\mathrm{R,s}}$. As $c_\mathrm{s} v_{\mathrm{R,s}}$ is an equilibrium of the system, we can write the distance dynamics as follows 
\be
\dot{\delta}(t) = \mathcal{Q}_k \delta(t) \, .
\ee
We cannot make a general statement about the definiteness of asymmetric matrices with non-positive eigenvalues. For the matrices belonging to $M_{v_\mathrm{s}}^\eqref{eqn:P1}$ or $M_{v_\mathrm{s}}^\eqref{eqn:P2}$ that are negative definite, we can bound the convergence speed of the distance vector as follows.
\begin{Corollary}
Consider a strongly connected graph for which the weighted in-degree or out-degree Laplacian is negative definite. For these networks, the distance vector globally asymptotically converges to zero and the convergence speed can be bounded as
\be
\| \delta(t)\|_2 \leq \|\delta(0)\|_2 \exp(q_\mathrm{max} t) \,
\ee 
where $q_\mathrm{max} = \max q_2(M_{v_\mathrm{s}}^\eqref{eqn:P1})$ for conservative networks and $q_\mathrm{max} = \max q_2(M_{v_\mathrm{s}}^\eqref{eqn:P2})$ for non-conservative networks. Here, $q_2$ is the non-zero eigenvalue with the smallest absolute value of the corresponding transition rate matrix.
\end{Corollary}
\begin{proof}
The proof of this corollary builds on the Lyapunov function $V(\delta) = 1/2 \delta(t)^T \delta(t)$. The first derivative can be written as 
\be
\dot{V}(\delta(t)) = 1/2 \left(\delta(t)^T \mathcal{Q}_k^T \delta(t) + \delta(t)^T \mathcal{Q}_k \delta(t) \right) \, ,
\ee
such that $\dot{V}(\delta(t)) \leq 0$ due to the assumptions made. Using the Courant-Fischer theorem, and similar to \cite{olfati2004consensus}, the proof is concluded. 
\end{proof}

\begin{Remark}
A relevant question is if it is possible to verify if two matrices $\mathcal{Q}_1$ and $\mathcal{Q}_2$ belong to the same set $M_{v_s}^{(\text{x})}$, $\text{x} \in \lbrace \ref{eqn:P1}, \ref{eqn:P2} \rbrace$. The common eigenvector problem can be solved when the eigenvalue is not known \cite{shemesh1984common, george1999common}. For the special case where the common eigenvalue is equal to zero, two matrices $\mathcal{Q}_1$ and $\mathcal{Q}_2$ have the same eigenvector corresponding to $q_s = 0$ if 
\be
\mathrm{null}(\mathcal{Q}_1) = \mathrm{null}(\mathcal{Q}_2) \, .
\ee
\end{Remark}
\section{Inhomogeneous equations: Stability and convergence}
\label{sec:inhomoSC}

To model the exogenous addition and subtraction of property to multi-agent networks, we extend the homogeneous equation \eqref{eqn:markov2} to include an inhomogeneous term 
\begin{equation}
\dot{\bar{\mathcal{S}}}(t) = \mathcal{Q}\bar{\mathcal{S}}(t) + \mathcal{U}(t) \, ,
\label{eqn:inhomoMaster}
\end{equation}
where $\mathcal{U}$ is an arbitrary input vector. %Using this input term, we can model point disturbances, sustained disturbances, and periodic disturbances at each node by means of Dirac delta functions, Heaviside step functions, and sinusoidal functions respectively. 
Since the inhomogeneous equation \eqref{eqn:inhomoMaster} is linear in the exogenous input term, the exogenous input can either be constant over all sample paths or averaged over all sample paths. Both interpretations will result in identical instance-averaged inputs. For linear, time-invariant systems, the system response can be derived in closed form. Note that the solution to \eqref{eqn:markov2} with initial conditions of the network state $\mathcal{S}(0)$ is given by $\bar{\mathcal{S}}_h(t) = \exp(\mathcal{Q} t)\mathcal{S}(0)$. It can be shown that the general solution to \eqref{eqn:inhomoMaster} is given by \cite{dahleh2004lectures}
\begin{equation}
\bar{\mathcal{S}}(t) = \bar{\mathcal{S}}_h(t) + \int_0^t \exp(\mathcal{Q} (t-\tau))\mathcal{U}(\tau)d\tau.
\label{eqn:inhomoMasterSoln}
\end{equation}
This can be interpreted as the sum of the solution to \eqref{eqn:markov2} and the convolution of the input $\mathcal{U}$ with the impulse response to \eqref{eqn:inhomoMaster} for every node in the network. 

We now discuss several important use cases for the input vector $\mathcal{U}$ and show that many scenarios can be described as in \eqref{eqn:inhomoMaster}. In particular, we will demonstrate that networks containing stubborn agents~\cite{acemoglu2013opinion} and networks with dynamic learning~\cite{acemoglu2008convergence} can both be expressed in the form of the inhomogeneous governing equation with constant and dynamic $\mathcal{U}$ respectively. This motivates the study of stability and convergence in the presence of exogenous inputs.

%%%%%%%%%%%%%%%%%%%%%%%%%%%%%%%%%%%%%%%%%%%%
\subsection{Non-conservative networks with static inputs}
Adopting the definition in~\cite{acemoglu2013opinion}, as well as the consensus update rule \eqref{eqn:P2}, a stubborn agent does not adjust its own value based on the values of neighboring nodes. This scenario is of interest to model opinion dynamics where a set of agents has constant opinion. In these networks, the out-degree of a stubborn agent is zero and the corresponding row of $\mathcal{Q}$ is a zero-row.\footnote{Conversely, if the in-degree of an agent equals zero, the node acts as a sink, which can receive property from every neighbor but does not contribute to his neighborhood. This case, where $\mathcal{Q}$ contains a zero column, does not lead to a convenient reduction of the state space.} For example, in the presence of two stubborn agents $a_1 = 2$ and $a_2 = n$, the governing equation can be written as follows
\begin{equation}
\begin{bmatrix}
\dot{\bar{\mathcal{S}}}_1 \\
\vdots \\
\dot{\bar{\mathcal{S}}}_n
\end{bmatrix}
= 
\begin{bmatrix}
\mathcal{Q}_{1,1}& \mathcal{Q}_{1,2}& \cdots & \mathcal{Q}_{1,n}\\
0 & 0 & \cdots & 0 \\
\mathcal{Q}_{3,1}& \mathcal{Q}_{3,2}& \cdots & \mathcal{Q}_{3,n}\\
\vdots & \vdots & \ddots & \vdots \\
0 & 0 & \cdots & 0 
\end{bmatrix}
\begin{bmatrix}
\bar{\mathcal{S}}_1 \\
\vdots\\
\bar{\mathcal{S}}_n
\end{bmatrix} \, ,
\end{equation}
which can be reformulated as
\begin{multline}
\label{eqn:stuAge}
\begin{bmatrix}
\dot{\bar{\mathcal{S}}}'_1 \\
\vdots \\
\dot{\bar{\mathcal{S}}}'_{n-2}
\end{bmatrix}
= 
\begin{bmatrix}
\mathcal{Q}'_{1,1}& \mathcal{Q}'_{1,2}& \cdots & \mathcal{Q}'_{1,n-2}\\
\vdots & \vdots & \ddots & \vdots \\
\mathcal{Q}'_{n-2,1}& \mathcal{Q}'_{n-2,2}& \cdots & \mathcal{Q}'_{n-2,n-2}
\end{bmatrix}
\begin{bmatrix}
\bar{\mathcal{S}}'_1 \\
\vdots\\
\bar{\mathcal{S}}'_{n-2}
\end{bmatrix}
\\+  
\begin{bmatrix}
\vdots & \vdots \\
B_{k,a_1}  & B_{k,a_2} \\
\vdots & \vdots
\end{bmatrix}
\begin{bmatrix}
\bar{\mathcal{S}}_{a_1}\\
\bar{\mathcal{S}}_{a_2}
\end{bmatrix} \, ,
\end{multline}
where $B_{k,a_1}$ and $B_{k,a_2}$ are the $a_1$-th and $a_2$-th columns of $\mathcal{Q}$ with the $a_1$-th and $a_2$-th entries omitted. The entries of these columns depend both on the network structure and the update profile, and we get
\begin{equation}
\label{eqn:inhConInp}
\dot{\bar{\mathcal{S}}}'(t) = \mathcal{Q}' \bar{\mathcal{S}}'(t) + B u \, ,
\end{equation}
where $\bar{\mathcal{S}}'$ and $\mathcal{Q}'$ represent the reduced state space and reduced transition rate matrix, and $u = [\bar{\mathcal{S}}_{a_1} \cdots \bar{\mathcal{S}}_{a_n}]^T$ contains the static quantities of the stubborn agents. We notice from \eqref{eqn:inhConInp} that the presence of the stubborn agents in the network leads to a reduction of the state space and a formulation equivalent to \eqref{eqn:inhomoMaster} with constant but possibly different inputs for different stubborn nodes, and $\mathcal{U} = Bu$. The reduced state space excludes the stubborn agents and the reduced transition rate matrix excludes the rows and columns corresponding to these agents. Note that the inhomogeneous equation \eqref{eqn:inhomoMaster} with static inputs corresponds to typical open-loop control systems.

Since the input vector $\mathcal{U}$ is constant, it is known that bounded-input bounded-output (BIBO) stability of \eqref{eqn:inhConInp} is guaranteed iff all eigenvalues of $\mathcal{Q}'$ have negative real components. In addition, since $\mathcal{Q}'$ is constructed differently from $\mathcal{Q}$, the convergence of the homogeneous and inhomogeneous systems is not identical. Define $\mathcal{I}_{\mathrm{st}} = \lbrace i \colon \bar{\mathcal{S}}_i(t) \, \text{is constant} \rbrace$ as the index set of stubborn agents. Then, we can write $\mathcal{Q}_{ii} = -\sum_{j \neq i, j \notin \mathcal{I}_{\mathrm{st}}} \mathcal{Q}_{ij} - \sum_{j \neq i, j \in \mathcal{I}_{\mathrm{st}}} \mathcal{Q}_{ij}$, and the row sum of $\mathcal{Q}'$ satisfies 

\begin{equation}
\mathcal{Q}'_{ii} + \sum_{j \neq i} \mathcal{Q}'_{ij} \leq 0 \, .
\end{equation}
As a consequence, $\mathcal{Q}'$ does not have a zero eigenvalue in general and is nonsingular in general. In the following Lemma, we provide a condition that guarantees the non-singularity of $\mathcal{Q}'$. 
\begin{Lemma}\label{lem:stuAgeSta}
In a network with stubborn agents, the transition rate matrix $\mathcal{Q'}$ of the reduced state space is invertible if all nodes are connected to at least one stubborn agent.
\end{Lemma}
\begin{proof}
If each agent is connected to at least one stubborn agent, $\mathcal{Q}'$ is strictly diagonally dominant and we can write
\be
\mathcal{Q}'_{ii} + \sum_{j \neq i} \mathcal{Q}'_{ij} < 0 \, .
\ee
According to Ger\v{s}gorin's circle theorem, all eigenvalues of $\mathcal{Q}'$ reside in the complex plane within the union of the disks $D_i = \lbrace z \in \mathbb{C} \,: \,  |z - \mathcal{Q}'_{ii}| \leq \sum_{j \neq i} |\mathcal{Q}'_{ij}|\rbrace \, , \, 1\leq i \leq n'$ with $n'$ the dimension of $\mathcal{Q}'$. As  $|\mathcal{Q}'_{ii}| > \sum_{j \neq i} |\mathcal{Q}'_{ij}|$ for a strictly diagonally dominant matrix, $\mathcal{Q}'$ cannot have a zero eigenvalue and is invertible.
\end{proof}
From Lemma \ref{lem:stuAgeSta}, we notice that for networks with stubborn agents the eigenvalues will on average be more concentrated around $\mathcal{Q}'_{ii}$, which can accelerate the network dynamics. For the convergence of a BIBO-stable reduced system involving $\mathcal{Q}'$ with the conditions of Lemma \ref{lem:stuAgeSta} satisfied, we formulate the following Lemma.
\begin{Lemma}
Let $Bu \doteq b$. If each agent is connected to at least one stubborn agent, then the network state converges to $\bar{\mathcal{S}}'^* = -{\mathcal{Q}'}^{-1}b$. We can then express the system dynamics as a homogeneous equation in the deviation of the property of a node from its steady state $\bar{\mathcal{S}}'_\Delta = \bar{\mathcal{S}}' - \bar{\mathcal{S}}'^*$:
\begin{equation}
\dot{\bar{\mathcal{S}}}'_\Delta (t) = \mathcal{Q}' \bar{\mathcal{S}}'_\Delta (t) \, .
\end{equation}
\label{lem:stubbornConv}
\end{Lemma}
\begin{proof}
The proof of this Lemma follows from traditional convergence results similar to those for the basic consensus algorithm~\cite{olfati2004consensus,olfati2007consensus,acemoglu2008convergence,acemoglu2013opinion} since the inhomogeneous equation in $\bar{\mathcal{S}}'$ can be expressed as a homogeneous equation in $\bar{\mathcal{S}}'_\Delta (t)$ and the poles of $\mathcal{Q}'$ lie in the open left-half complex plane.

Global asymptotic stability can be demonstrated for symmetric $\mathcal{Q'}$ by using the standard Lyapunov function $V = \bar{\mathcal{S}}_\Delta^{'T} \bar{\mathcal{S}}'_\Delta$. In this case, we obtain
\begin{equation}
\dot{V} = \bar{\mathcal{S}}_\Delta^{'T} \mathcal{Q}'^T \bar{\mathcal{S}}'_\Delta + \bar{\mathcal{S}}_\Delta^{'T} \mathcal{Q}' \bar{\mathcal{S}}'_\Delta < 0 \enskip \enskip \forall  \,\, \bar{\mathcal{S}}'_\Delta \neq \vec{0} \notag
\end{equation}
since $\mathcal{Q}'$ and thus its transpose are negative definite. This demonstrates that the equilibrium point of the network is globally asymptotically stable, and the location of the point can be obtained by setting $\dot{\bar{\mathcal{S}}}'_\Delta=\vec{0}$. 

The Lyapunov function above does not hold for nonsymmetric $\mathcal{Q}'$. However, because $\mathcal{Q}'$ is strictly diagonally dominant, the existence of a Lyapunov function for $\mathcal{Q}'$ can be demonstrated to prove the global asymptotic stability of $\bar{\mathcal{S}}'$~\cite{willems1976lyapunov}. One such Lyapunov function is $ V = \max_i |\bar{\mathcal{S}}'_{\Delta, i}|$, which proves the global asymptotic stability of $\bar{\mathcal{S}}'$.
\end{proof}

Note that in the general case of conservative or non-conservative updating without stubborn agents, $\mathcal{Q}$ is constructed such that it has a steady state eigenvalue $q_\mathrm{s} = 0$. As a consequence, $\mathcal{Q}$ is not invertible. If $\mathcal{Q}$ (or $\mathcal{Q}'$) is singular, then BIBO stability is no longer guaranteed, and the system is instead only marginally stable.

\begin{Lemma}
A non-trivially convergent network where $q_\mathrm{s} = 0$ exists, diverges in infinite time in the presence of a constant input vector $b$.
\end{Lemma}
\begin{proof}
Consider the convolution integral in \eqref{eqn:inhomoMasterSoln}. Since $\mathcal{Q}$ is not invertible, we perform the convolution integral by expanding the matrix exponential in a power series and performing an eigendecomposition of $\mathcal{Q}$. Then, since the steady-state eigenvector does not decay with time, the integral does not converge in infinite time.
\end{proof}

%%%%%%%%%%%%%%%%%%%%%%%%%%%%%%%%%%%%%%%%%%%%
\subsection{Non-conservative networks with dynamic inputs}
In the case of dynamic learning, each node updates its own value based on the value of its neighboring nodes, as well as the difference of its own value with a measurement of the system state~\cite{acemoglu2008convergence}, which can be interpreted as a reference input for the node. This scenario is of interest to model tracking systems that are augmented with the information present in the neighborhood. This can be represented by the following update rule

\begin{equation}
\mathcal{S}_i(t+\Delta t) = \mathcal{S}_i(t) + \sum_{j\neq i} \mathcal{C}_{ij} (\mathcal{S}_j(t)-\mathcal{S}_i(t)) + \beta(t) (X_i(t)-\mathcal{S}_i(t)) \, ,
\label{eqn:dynamicLearning}
\end{equation}
where $X_i(t)$ represents the measurement of node $i$ at time $t$. If the measurement of each node is made at a rate $\rho$, then the corresponding governing equation is
\begin{equation}
\dot{\bar{\mathcal{S}}}(t) = (\mathcal{Q}-\beta'\mathbb{I})\bar{\mathcal{S}}(t) + \beta' X(t) \, ,
\label{eqn:dynamicInput}
\end{equation}
where $\beta' = \beta \rho$, and where $r_{ij} = \rho$ in the construction of $Q_{ij}$. This is again reminiscent of \eqref{eqn:inhomoMaster} with time-varying $\mathcal{U} = \beta' X(t)$. The inhomogeneous governing equation with dynamic reference inputs corresponds directly to typical closed-loop proportional control systems with reference input $X$ and proportional gain $\beta'$.

Analogous to the case of systems with stubborn agents above is the BIBO stability criteria for the system with dynamic learning and constant $\beta'$: iff all eigenvalues of $\mathcal{Q}_D = \mathcal{Q}-\beta'\mathbb{I}$ have negative real components, then an input $\beta' X(t)$ that is bounded for all time $t>0$ will result in BIBO stability. In addition, we formulate the following Lemmas for the convergence of a BIBO stable dynamic learning system with constant $\beta'$.

\begin{Lemma}
If $\mathcal{Q}_D$ is invertible and $\lim_{t \rightarrow \infty} X(t) = X^*$ for some real-valued constant $X^*$, then the network state converges to $\bar{\mathcal{S}}^* = -{\mathcal{Q}_D}^{-1} \beta' X^*$.
\label{lem:dynamicConv}
\end{Lemma}
\begin{proof}
The proof of this Lemma is analogous to the proof of Lemma \ref{lem:stubbornConv}. 
\end{proof}

\begin{Remark}
If $\mathcal{Q}$ corresponds to a CTMC and $\beta' > 0$, then the eigenvalues of $\mathcal{Q}_D$ are exactly $\beta'$ less than the eigenvalues of $\mathcal{Q}$, and the resultant system is guaranteed to be BIBO stable, provided $\beta'$ takes the same value for all agents.
\end{Remark}

In the more general case where $X(t)$ is not an independent measurement, but a function of $\bar{\mathcal{S}}_\Delta$, we can formulate the following Lemma.

\begin{Lemma}
If $\bar{\mathcal{S}}_\Delta$ reaches equilibrium much faster than $X$, then we can approximate $\dot{X}(t) = h(X,\bar{\mathcal{S}}_\Delta)$ as $\dot{X}(t) \approx h(X,0)$ for any arbitrary function $h$ and regardless of whether $X$ eventually converges. The time evolutions of $\bar{\mathcal{S}}_\Delta$ and $X$ can then be approximated as independent even if they were originally coupled.
\end{Lemma}
\begin{proof}
This Lemma follows directly from the moving equilibrium theorem.
\end{proof}

%%%%%%%%%%%%%%%%%%%%%%%%%%%%%%%%%%%%%%%%%%%%
\subsection{Dynamic non-reference inputs}
Consider the following variation of \eqref{eqn:dynamicInput}

\begin{equation}
\dot{\bar{\mathcal{S}}}(t) = \mathcal{Q}\bar{\mathcal{S}}(t) + \mathcal{U}(t) \, ,
\end{equation}
where $\mathcal{U}(t)$ is some arbitrary time-varying input vector and $\mathcal{Q}$ can be used to described a CTMC. The inhomogeneous governing equation with dynamic non-reference inputs corresponds to typical open-loop control systems. Since $\mathcal{Q}$ is singular, the system is only marginally stable and not BIBO stable. A marginally stable system is known to give a bounded output if its input is an impulse of finite magnitude. By the principle of superposition, a continuously finite input that eventually decays to zero will also give a bounded output.

%%%%%%%%%%%%%%%%%%%%%%%%%%%%%%%%%%%%%%%%%%%%
\subsection{Expanded state space}
Consider the following variation of \eqref{eqn:dynamicLearning}
\begin{multline}
\mathcal{S}_i(t+\Delta t) = \mathcal{S}_i(t) + \sum_{j\neq i} \mathcal{C}_{ij} (\mathcal{S}_j(t)-\mathcal{S}_i(t)) + \beta_1(X_i(t)-\mathcal{S}_i(t)) \\
+ \beta_2(\dot{X}_i(t)-\dot{\mathcal{S}}_i(t)) + \beta_3(Y_i(t)-\mathcal{T}_i(t)) \, ,
\end{multline}
where $\dot{Y}_i(t) = X_i(t)$ and $\dot{\mathcal{T}}_i(t) = \mathcal{S}_i(t)$. This leads us to the coupled equations
\begin{align}
\dot{\bar{\mathcal{T}}}(t) &= \bar{\mathcal{S}}(t) \, , \\
(1+\beta'_2)\dot{\bar{\mathcal{S}}}(t) &=
\begin{aligned}[t]
& (\mathcal{Q}-\beta'_1\mathbb{I})\bar{\mathcal{S}}(t) - \beta'_3\bar{\mathcal{T}}(t) \\
& + \beta'_1 X(t) + \beta'_2 \dot{X}(t) + \beta'_3 Y(t) \, ,
\end{aligned}
\end{align}
where $\beta'_1 = \beta_1 \rho$, $\beta'_2 = \beta_2 \rho$ and $\beta'_3 = \beta_3 \rho$. The inhomogeneous governing equation with an expanded state space corresponds to typical closed-loop proportional, integral and derivative (PID) control systems. Here, $\beta_1$, $\beta_3$ and $\beta_2$ correspond to P, I and D controls respectively. The expanded state space now includes the components of both $\mathcal{S}$ and $\mathcal{T}$, and we can write this as a single relation
\begin{multline}
(1 + \beta'_2) \mathbb{I} \ddot{\bar{\mathcal{T}}}(t) + (\beta'_1 \mathbb{I} - \mathcal{Q}) \dot{\bar{\mathcal{T}}}(t) + \beta'_3 \mathbb{I} \bar{\mathcal{T}}(t) \\
= \beta'_1 \dot{Y}(t) + \beta'_2 \ddot{Y}(t) + \beta'_3 Y(t) \, .
\end{multline}
Performing a Laplace transform and then rearranging, we obtain
\be
\dfrac{\hat{\bar{\mathcal{S}}}}{\hat{X}} = 
(\beta'_2 s^2 + \beta'_1 s + \beta'_3) \left[ (1+\beta'_2)\mathbb{I} s^2 + (\beta'_1 \mathbb{I} - \mathcal{Q}) s + \beta'_3 \mathbb{I} \right]^{-1}   \, ,
\label{eqn:PIDTransfer}
\ee
where the hat above a quantity denotes a Laplace-transformed quantity. The stability of this system is determined by the roots of the characteristic polynomial of the system, which can be obtained from the determinant of the matrix $\mathcal{Q}_E = \left[ (1+\beta'_2)\mathbb{I} s^2 + (\beta'_1 \mathbb{I} - \mathcal{Q}) s + \beta'_3 \mathbb{I} \right]$. From standard state-space theory, if all the poles of the system, or the roots of the characteristic polynomial of $\mathcal{Q}_E$, lie in the open left-half complex plane, then the system is BIBO stable. If at least one of the poles lies on the complex axis, then the system is instead marginally stable. Depending on the values of $\beta'_1$, $\beta'_2$ and $\beta'_3$, the convergence of a BIBO stable system can be derived. For example, if $\beta'_2 = \beta'_3 = 0$, we recover the results of Lemma \ref{lem:dynamicConv}; on the other hand, if $\beta'_3 \neq 0$, the final value theorem tells us that taking the limit $s \rightarrow 0$ of $s\hat{\bar{\mathcal{S}}}$ gives us the steady state value $\hat{\bar{\mathcal{S}}}_\infty = \hat{X}_\infty$.

\section{Control by exogenous excitation}
\label{sec:External}
Whereas we studied in Section \ref{sec:inhomoSC} the convergence and stability under exogenous inputs, we demonstrate in this section how the rates of property diffusion can be controlled through appropriate design of a generic time-varying input vector $\mathcal{U}(t)$. 

\subsection{Using inputs for targeted control}

We first establish a methodology for decomposing the existing system dynamics into a more convenient form for targeted control of individual system modes. We summarize this decomposition, as well as the subsequent reconstruction for implementation in real networks, in Algorithm \ref{alg:inputTransform}.

\begin{algorithm}
\caption{Implementing inputs for targeted control}
\label{alg:inputTransform}
\begin{algorithmic}[1]
\STATE Transform basis of nodes to quasi-node space using characteristic eigenvectors as in \eqref{eqn:transformBasis}
\STATE Perform a Laplace transform of the system to move into the frequency domain
\IF {implementation of external controllers is feasible}
	\STATE Implement controllers on one or more quasi-modes to achieve control objective
\ELSE
	\STATE Design quasi-inputs by incorporating controller action on one quasi-mode
\ENDIF
\STATE Perform an inverse Laplace transform back into the time domain
\STATE Transform quasi-nodes back to node space to obtain desired node inputs or system dynamics
\end{algorithmic}
\end{algorithm}

\begin{figure}[!htb]
\centering
\begin{tikzpicture}[auto, node distance=2cm,>=latex',every node/.style={scale=0.65}]
    \node [input, name=rinput] (rinput) {};
    \node [sum, right of=rinput] (sum1) {};
    \node [block, right of=sum1] (controller) {$\dfrac{1}{s}$};    
    \node [sum, right of=controller, node distance=1.5cm] (sum2) {};
    \node [input, above=0.27cm of sum2] (initial) {};
    \node [output, right of=sum2] (output) {};
    \node [block, below=0.13cm of controller] (feedback) {$-\mathcal{Q}$};
    
    \draw [->] (rinput) -- node{$\hat{\mathcal{U}}$} (sum1);
    \draw [->] (sum1) -- node{$\hat{\dot{\bar{\mathcal{S}}}}$} (controller);
    \draw [->] (controller) -- (sum2);
    \draw [->] (sum2) -- node [name=y] {$\hat{\bar{\mathcal{S}}}$}(output);
    \draw [->] (y) |- (feedback);
    \draw [->] (feedback) -| node[pos=0.99] {$-$} (sum1);
    \draw [->] (initial)node[above]{$\mathcal{S}(0)$} -- node[pos=0.99] {$+$} (sum2); 
\end{tikzpicture}
\caption{Block diagram for the inhomogeneous equation \eqref{eqn:inhomoMaster}. The hat above a quantity denotes a Laplace-transformed quantity.} \label{fig:block1}
\end{figure}
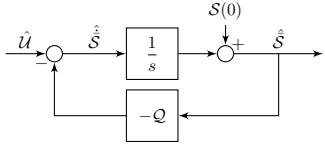
Figure \ref{fig:block1} demonstrates that \eqref{eqn:inhomoMaster} together with initial conditions $\mathcal{S}(0)$ can be expressed as a block diagram with a feedback loop. It is possible to separate the effects of the initial conditions and the external input, which is illustrated in Fig. \ref{fig:block2} for clarity. In what follows, we study the effects of modifying the lower branch in Fig. \ref{fig:block2} related to the exogenous input without influencing the upper branch. 
\begin{figure}[!htb]
\centering
\begin{tikzpicture}[auto, node distance=2cm,>=latex',every node/.style={scale=0.65}]
    \node [input, name=rinput1] (rinput1) {};
    \node [block, right of=rinput1] (controller1) {$s(s\mathbb{I}-\mathcal{Q})^{-1}$};
    \node [sum, below=0.45cm of rinput1, right of=controller1, node distance=2cm] (sum1) {};
    \node [block, below=0.45cm of sum1, left of=sum1, node distance=2cm] (controller2) {$(s\mathbb{I}-\mathcal{Q})^{-1}$};
    \node [input, left of=controller2] (rinput2) {};
    \node [output, right of=sum1] (output) {};
    
    \draw [->] (rinput1) -- node{$\mathcal{S}(0)$} (controller1);
    \draw [->] (controller1) -| node[pos=0.99] {$+$} (sum1);
    \draw [->] (controller2) -| node[pos=0.99] {$+$} (sum1);
    \draw [->] (rinput2) -- node{$\hat{\mathcal{U}}$} (controller2);
    \draw [->] (sum1) -- node {$\hat{\bar{\mathcal{S}}}$} (output);     
\end{tikzpicture}
\caption{Simplified block diagram for (\ref{eqn:inhomoMaster}).} \label{fig:block2}
\end{figure}
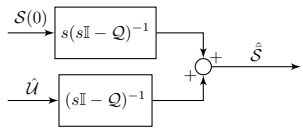
We consider here the case where the system has exactly $n$ distinct eigenvalues and eigenvectors. We use the decomposition in \eqref{eqn:eigenMaster} and rewrite \eqref{eqn:inhomoMaster} as 
\begin{equation}
A^{-1} \dot{\bar{\mathcal{S}}} = \Lambda A^{-1}\bar{\mathcal{S}} + A^{-1}\mathcal{U} \, ,
\label{eqn:transformBasis}
\end{equation}
which corresponds to performing a basis transformation by taking linear combinations of the various nodes. Under this transformation, we obtain \emph{quasi-modes} with quasi-property $\tilde{s}_i$, as well as corresponding \emph{quasi-inputs} $\tilde{u}_i$ that are each acted upon by a single eigenvalue $q_i$ of $\mathcal{Q}$. The time evolution for this transformed system is illustrated in Fig. \ref{fig:block3}.
\begin{figure}[!htb]
\centering
\begin{tikzpicture}[auto, node distance=2cm,>=latex',every node/.style={scale=0.65}]
    \node [input, name=rinput1] (rinput1) {};
    \node [block, right of=rinput1] (controller1) {$\dfrac{s}{s-q_i}$};
    \node [sum, below=0.45cm of rinput1, right of=controller1, node distance=2cm] (sum1) {};
    \node [block, below=0.45cm of sum1, left of=sum1, node distance=2cm] (controller2) {$\dfrac{1}{s-q_i}$};
    \node [input, left of=controller2] (rinput2) {};
    \node [output, right of=sum1] (output) {};
    
    \draw [->] (rinput1) -- node{$\tilde{s}_i(0)$} (controller1);
    \draw [->] (controller1) -| node[pos=0.99] {$+$} (sum1);
    \draw [->] (controller2) -| node[pos=0.99] {$+$} (sum1);
    \draw [->] (rinput2) -- node{$\hat{\tilde{u}}_i$} (controller2);
    \draw [->] (sum1) -- node {$\hat{\tilde{\bar{s}}}_i$} (output);
       
\end{tikzpicture}
\caption{Simplified block diagram for the inhomogeneous transformed governing equation for a single quasi-mode with quasi-property $\dot{\tilde{\bar{s}}}_i = q_i\tilde{\bar{s}}_i + \tilde{u}_i$. The tilde above a quantity and its lower case denote a quantity in quasi-node space.} \label{fig:block3}
\end{figure}
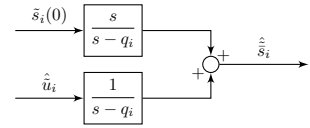

After applying the basis transformation, we can tune the quasi-modes individually to achieve the desired control objective using one of the two following approaches. We can either design a controller acting on one or more of the quasi-modes to steer the system response, or modify the quasi input such that it incorporates the controller action. In the first approach, we design a closed-loop feedback control block to adjust the system response to different input signals. These include the implementation of filters and lead-lag compensators to modify the frequency and phase responses. Conversely, we can design the quasi-input $\tilde{u}'_i=\tilde{u}'(\tilde{u}_i)$ for all quasi-modes $i$ in order to bring the quasi-outputs to their desired values. This open-loop approach could be convenient in scenarios where perfect measurements of the system state are difficult to make and the effects of system noise are not significant. We illustrate this with a simple example in Fig. \ref{fig:quasiinput2}. Suppose we wish to mimic proportional or integral closed-loop control on the quasi-input $\tilde{u}_y$ corresponding to quasi-mode $y$ using a feedback block $F$ with transfer function $K$ or $K/s$ respectively, where $K \in \mathbb{R}$. We then subsume the control loop into a quasi-input using the Laplace-transformed relation
\begin{equation}
\hat{\tilde{u}}'(\hat{\tilde{u}}_i) = \hat{\tilde{u}}_i\dfrac{s-q_y}{s-q_y+F},
\label{eq:netquasiinput}
\end{equation}
which makes the two block diagrams in Fig. \ref{fig:quasiinput2} equivalent. Note that the controller design here is only performed on a single quasi-mode $y$. Yet, the quasi-inputs of all the other quasi-modes must be transformed using \eqref{eq:netquasiinput}, since all the quasi-modes are subject to the same external input. The main advantage of the proposed method lies in the possibility to isolate parts of the system and separately adjust the response times of the different modes without a physical controller. 

\begin{figure}[!htb]
\centering
\subfloat{
\begin{tikzpicture}[auto, node distance=2cm,>=latex',every node/.style={scale=0.65}]
    \node [input, name=rinput] (rinput) {};
    \node [sum, right of=rinput, node distance=1cm] (sum1) {};
    \node [block, right of=sum1, node distance=1.5cm] (controller) {$\dfrac{1}{s-q_y}$};    
    \node [output, right of=controller, node distance=2cm] (output) {};
    \node [block, below of=controller, node distance=1.2cm] (feedback) {F};
    
    \draw [->] (rinput) -- node{$\hat{\tilde{u}}_y$} (sum1);
    \draw [->] (sum1) -- (controller);
    \draw [->] (controller) -- node [name=y] {$\hat{\tilde{\bar{s}}}_y$}(output);
    \draw [->] (y) |- (feedback);
    \draw [->] (feedback) -| node[pos=0.99] {$-$} (sum1);
        
\end{tikzpicture}
}
\mbox{\enskip}\mbox{\enskip}\mbox{\enskip}\mbox{\enskip}
\subfloat{
\begin{tikzpicture}[auto, node distance=2cm,>=latex',every node/.style={scale=0.65}]
    \node [input, name=rinput] (rinput) {};
    \node [block, right of=rinput, node distance=2cm] (controller) {$\dfrac{1}{s-q_y}$};    
    \node [output, right of=controller, node distance=2cm] (output) {};
    
    \draw [->] (rinput) -- node{$\hat{\tilde{u}}'(\hat{\tilde{u}}_y)$} (controller);
    \draw [->] (controller) -- node {$\hat{\tilde{\bar{s}}}_y$}(output);
        
\end{tikzpicture}
}
\caption{Block diagrams for quasi-mode $y$ for the original control loop (left) and when the loop is subsumed into the Laplace-transformed quasi-input $\hat{\tilde{u}}'(\hat{\tilde{u}}_i)$ (right).} \label{fig:quasiinput2}
\end{figure}
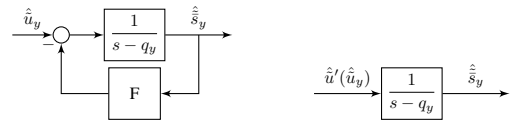

\begin{Remark}
Note that it is possible to account for the case where some of the $n$ eigenvalues are degenerate as long as $\mathcal{Q}$ is diagonalizable. In the event that the algebraic and geometric multiplicities of each of the degenerate eigenvalues are identical, we can use the same transformation as above, except that some of the quasi-modes corresponding to the same eigenvalue, yet different eigenvectors, will have the same transfer function $1/(s-q_i)$. The situation becomes more sensitive when $\mathcal{Q}$ is defective. In this case, we have to consider the generalized eigenspace of the matrix instead. Using the generalized eigenvector decomposition of $\mathcal{Q}$, we obtain a complete basis of $n$ orthogonal generalized eigenvectors, and $\mathcal{Q}$ can be diagonalized in the block Jordan form. Suppose we have a Jordan block involving nodes $z$ and $z+1$. Then, $\dot{\tilde{\bar{s}}}_z = q_z \tilde{\bar{s}}_z + \tilde{\bar{s}}_{z+1} + \tilde{u}_z$. %The block diagram corresponding to this is illustrated in Fig. \ref{fig:quasiinput3}.
\end{Remark}

%\begin{figure}[!htb]
%\centering
%\begin{tikzpicture}[auto, node distance=2cm,>=latex',every node/.style={scale=0.65}]
%    \node [input, name=rinput1] (rinput1) {};
%    \node [block, right of=rinput1] (controller1) {$\dfrac{s}{s-q_z}$};
%    \node [sum, below=0.9cm of rinput1, right of=controller1, node distance=2cm] (sum1) {};
%    \node [block, left of=sum1, node distance=2cm] (controller2) {$\dfrac{1}{s-q_z}$};
%    \node [block, below=0.9cm of sum1, left of=sum1, node distance=2cm] (controller3) {$\dfrac{1}{s-q_z}$};
%    \node [input, left of=controller2] (rinput2) {};
%    \node [input, left of=controller3] (rinput3) {};
%    \node [output, right of=sum1] (output) {};
%    
%    \draw [->] (rinput1) -- node{$\hat{\tilde{s}}_z(0)$} (controller1);
%    \draw [->] (controller1) -| node[pos=0.99] {$+$} (sum1);
%    \draw [->] (controller2) -- node[pos=0.99] {} (sum1);
%    \draw [->] (controller3) -| node[pos=0.99] {} (sum1);
%    \draw [->] (rinput2) -- node{$\hat{\tilde{u}}_z$} (controller2);
%    \draw [->] (rinput3) -- node{$\hat{\tilde{\bar{s}}}_{z+1}$} (controller3);
%    \draw [->] (sum1) -- node {$\hat{\tilde{\bar{s}}}_z$} (output);
%       
%\end{tikzpicture}
%\caption{Block diagram for the degenerate quasi-mode $\dot{\tilde{\bar{s}}}_z = q_z \tilde{\bar{s}}_z + \tilde{\bar{s}}_{z+1} + \tilde{u}_z$. The $z$-th and $(z+1)$-th quasi-modes are obtained after transformation with the $z$-th and $(z+1)$-th generalized eigenvectors respectively.} \label{fig:quasiinput3}
%\end{figure}
%

\subsection{Case study for external control}
\label{sec:extCtrl}
\begin{figure}[!htb]
\centering
\begin{tikzpicture}[->,auto,node distance=2cm,>=latex',graph node/.style={circle,draw},every node/.style={scale=0.75}]
	
	\node [graph node] (1) {1};
	\node [graph node] (2) [right of=1] {2};
	\node [graph node] (3) [below of=2] {3};
	\node [graph node] (4) [left of=3] {4};
	
	\path (1) edge [bend left] node {1} (2);
	\path (2) edge [bend left] node {0.5} (3);
	\path (3) edge [bend left] node {1} (4);
	\path (4) edge [bend left] node {0.5} (1);
	\path (1) edge [bend left] node {1} (4);
	\path (4) edge [bend left] node {0.5} (3);
	\path (3) edge [bend left] node {1} (2);
	\path (2) edge [bend left] node [below] {0.5} (1);
	
\end{tikzpicture}
\caption{Asymmetric cycle graph $C_{4,a}$.} \label{fig:externalNetwork}
\end{figure}
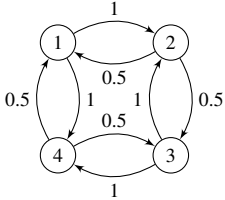
To illustrate how external control on a single quasi-mode can be implemented by quasi-input design, we consider a simple case study of four nodes in an asymmetric cycle graph configuration, as shown in Fig. \ref{fig:externalNetwork}. In this case, the eigenvalues of the associated $\mathcal{Q}$ are $\{q_1,q_2,q_3,q_4\} = \{-3,-2,-1,0\}$, and there are no degeneracies. 

Under \eqref{eqn:P1}, the corresponding unit right eigenvectors are $\tfrac{1}{2}\left[\begin{smallmatrix}-1 & 1 & -1 & 1\end{smallmatrix}\right]$, $\sqrt{\tfrac{1}{2}}\left[\begin{smallmatrix}-1 & 0 & 1 & 0\end{smallmatrix}\right]$, $\sqrt{\tfrac{1}{2}}\left[\begin{smallmatrix}0 & -1 & 0 & 1\end{smallmatrix}\right]$, and $-\sqrt{\tfrac{2}{5}}\left[\begin{smallmatrix}\tfrac{1}{2} & 1 & \tfrac{1}{2} & 1\end{smallmatrix}\right]$ respectively. Based on the unit steady-state eigenvector, we observe that the network on average tends to accumulate property at nodes 2 and 4 while depleting the property of the other nodes. This motivates us to consider the control objective of maintaining property at nodes 1 and 3 during the transient period to reverse the steady-state asymmetry. 

Under \eqref{eqn:P2}, the corresponding unit right eigenvectors are $\sqrt{\tfrac{2}{5}}\left[\begin{smallmatrix}-1 & \tfrac{1}{2} & -1 & \tfrac{1}{2}\end{smallmatrix}\right]$, $\sqrt{\tfrac{1}{2}}\left[\begin{smallmatrix}1 & 0 & -1 & 0\end{smallmatrix}\right]$, $\sqrt{\tfrac{1}{2}}\left[\begin{smallmatrix}0 & -1 & 0 & 1\end{smallmatrix}\right]$ and $\tfrac{1}{2}\left[\begin{smallmatrix}1 & 1 & 1 & 1\end{smallmatrix}\right]$ respectively. Based on the network structure, we observe that the network allocates higher weights to the quantities at nodes 2 and 4, which affects the calculation of the final consensus value. This motivates us to consider the control objective of increasing the weight of nodes 1 and 3 during the transient period to increase their influence on the system.

For either update protocol, it is sensible to excite nodes 1 and 3 in order to achieve the control objective. Since each quasi-mode is associated with a single pole, the net effect of the quasi-mode plant resembles the integration of the input, so we inject a unit impulse $\delta(t)$ into nodes 1 and 3 at $t=0$, and then apply control to displace the poles and achieve step-like outputs. The injection is equivalent to a Laplace-transformed input $\hat{\mathcal{U}} = \left[1 \, 0\,  1\, 0\right]$. By performing the aforementioned basis transformation on $\hat{\mathcal{U}}$, we obtain the equivalent quasi-inputs $\{\hat{\tilde{u}}_1,\hat{\tilde{u}}_2,\hat{\tilde{u}}_3,\hat{\tilde{u}}_4\}=\{-\tfrac{4}{3}, 0, 0, -2\sqrt{\tfrac{5}{18}}\}$ for the conservative network and $\{\hat{\tilde{u}}_1,\hat{\tilde{u}}_2,\hat{\tilde{u}}_3,\hat{\tilde{u}}_4\}=\{-2\sqrt{\tfrac{5}{18}}, 0, 0, \tfrac{2}{3}\}$ for the non-conservative network. Note that quasi-modes 2 and 3 generate zero quasi-inputs in both cases due to the asymmetrically designed input $\left[1 \, 0\,  1\, 0\right]$. Hence, from here on, we will only consider the quasi-modes 1 and 4 with non-trivial quasi-inputs for this particular case study. Figure \ref{fig:quasiinput} illustrates the block diagrams for these two quasi-modes in the $s$-domain, along with possible closed-loop control blocks $H_1$ and $H_4$. 

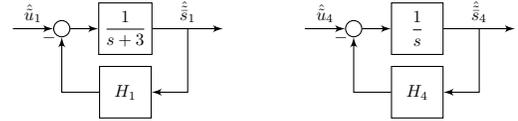
\begin{figure}[!htb]
\centering

\subfloat{
\begin{tikzpicture}[auto, node distance=2cm,>=latex',every node/.style={scale=0.65}]
    \node [input, name=rinput] (rinput) {};
    \node [sum, right of=rinput, node distance=1cm] (sum1) {};
	\node [block, right of=sum1, node distance=1.3cm] (controller) {$\dfrac{1}{s+3}$};    
    \node [output, right of=controller, node distance=2cm] (output) {};
    \node [block, below of=controller, node distance=1.3cm] (feedback) {$H_1$};
    
    \draw [->] (rinput) -- node{$\hat{\tilde{u}}_1$} (sum1);
    \draw [->] (sum1) -- (controller);
    \draw [->] (controller) -- node [name=y] {$\hat{\tilde{\bar{s}}}_1$}(output);
    \draw [->] (y) |- (feedback);
    \draw [->] (feedback) -| node[pos=0.99] {$-$} (sum1);
        
\end{tikzpicture}
}
\mbox{\enskip}\mbox{\enskip}\mbox{\enskip}\mbox{\enskip}
\subfloat{
\begin{tikzpicture}[auto, node distance=2cm,>=latex',every node/.style={scale=0.65}]
    \node [input, name=rinput] (rinput) {};
    \node [sum, right of=rinput, node distance=1cm] (sum1) {};
	\node [block, right of=sum1, node distance=1.3cm] (controller) {$\dfrac{1}{s}$};    
    \node [output, right of=controller, node distance=2cm] (output) {};
    \node [block, below of=controller, node distance=1.3cm] (feedback) {$H_4$};
    
    \draw [->] (rinput) -- node{$\hat{\tilde{u}}_4$} (sum1);
    \draw [->] (sum1) -- (controller);
    \draw [->] (controller) -- node [name=y] {$\hat{\tilde{\bar{s}}}_4$}(output);
    \draw [->] (y) |- (feedback);
    \draw [->] (feedback) -| node[pos=0.99] {$-$} (sum1);
        
\end{tikzpicture}
}
\caption{Block diagrams for quasi-modes 1 (left) and 4 (right) of the system illustrated in Fig. \ref{fig:externalNetwork}, with corresponding Laplace-transformed quasi-outputs $\hat{\tilde{\bar{s}}}_1$ and $\hat{\tilde{\bar{s}}}_4$. The controllers $H_1$ and $H_4$ represent a feedback control block.} \label{fig:quasiinput}
\end{figure}

\subsubsection{No control}
In the absence of controllers ($H_1 = H_4 = 0$), we obtain the following results:

Under \eqref{eqn:P1}, we obtain the quasi-outputs $\tilde{\bar{s}}_1 = -\tfrac{4}{3}e^{-3t}$ and $\tilde{\bar{s}}_4 = -2\sqrt{\tfrac{5}{18}}$. After transforming the quasi-outputs back into node space, we obtain $\bar{\mathcal{S}}_{1} = \bar{\mathcal{S}}_{3} = \tfrac{1}{3}(1+2e^{-3t})$ for the instance-averaged properties of nodes 1 and 3, and $\bar{\mathcal{S}}_{2} = \bar{\mathcal{S}}_{4} = \tfrac{2}{3}(1-e^{-3t})$ for the instance-averaged properties of nodes 2 and 4.

Under \eqref{eqn:P2}, we obtain the quasi-outputs $\tilde{\bar{s}}_1 = -2\sqrt{\tfrac{5}{18}}e^{-3t}$ and $\tilde{\bar{s}}_4 = \tfrac{2}{3}$, which then correspond to $\bar{\mathcal{S}}_{1} = \bar{\mathcal{S}}_{3} = \tfrac{1}{3}(1+2e^{-3t})$, and $\bar{\mathcal{S}}_{2} = \bar{\mathcal{S}}_{4} = \tfrac{1}{3}(1-e^{-3t})$. 

In both cases, the settling time of the outputs to their steady-state values is $1/3$. We can adjust this settling time using proportional control.

\subsubsection{Proportional control}

\begin{figure} [!t]
	\centering
	\subfloat [Conservative protocol]{
		\includegraphics[width=0.5\linewidth]{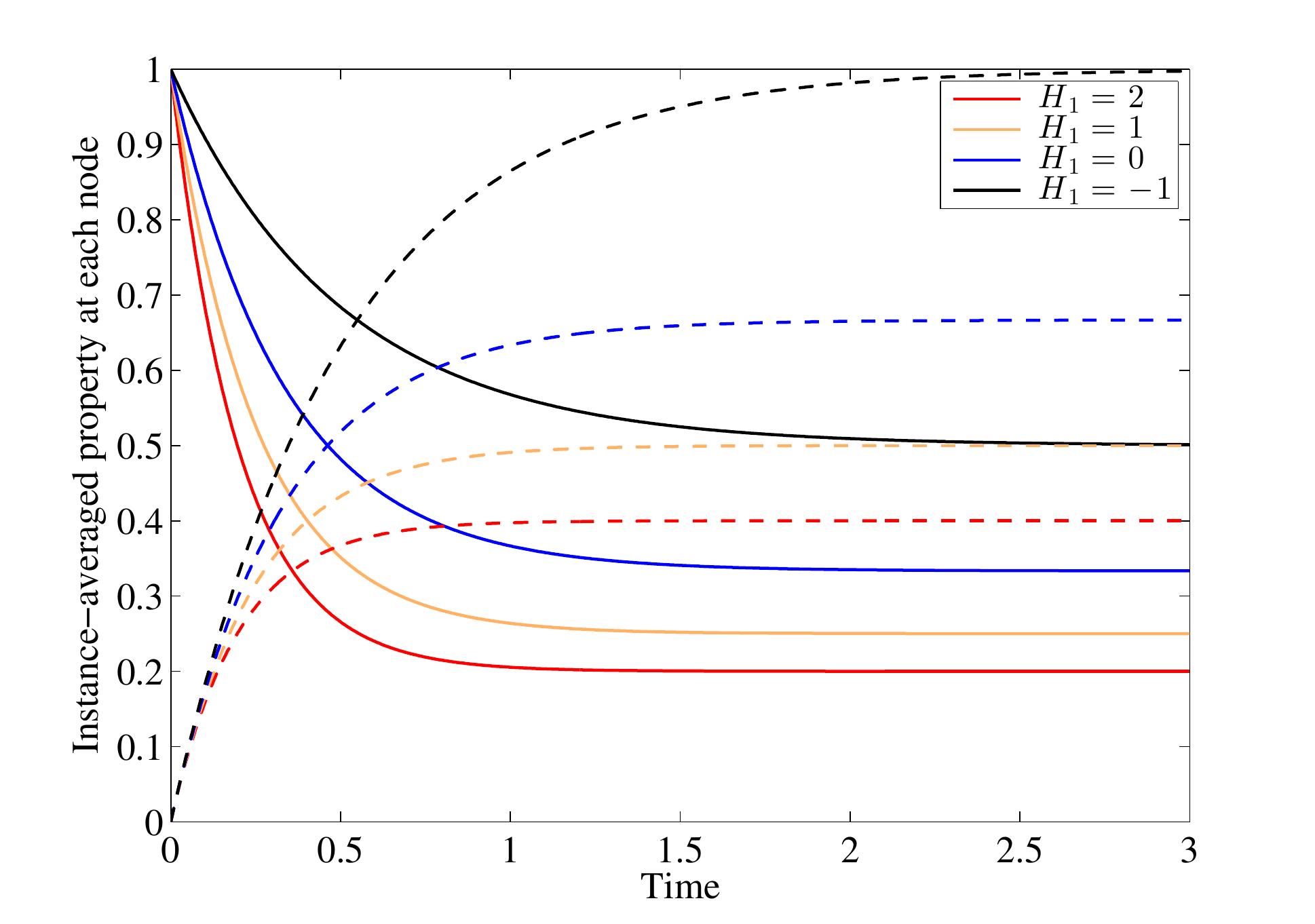}
		\label{fig:propCtrlConserv}
	}
	\subfloat [Non-conservative protocol]{
		\includegraphics[width=0.5\linewidth]{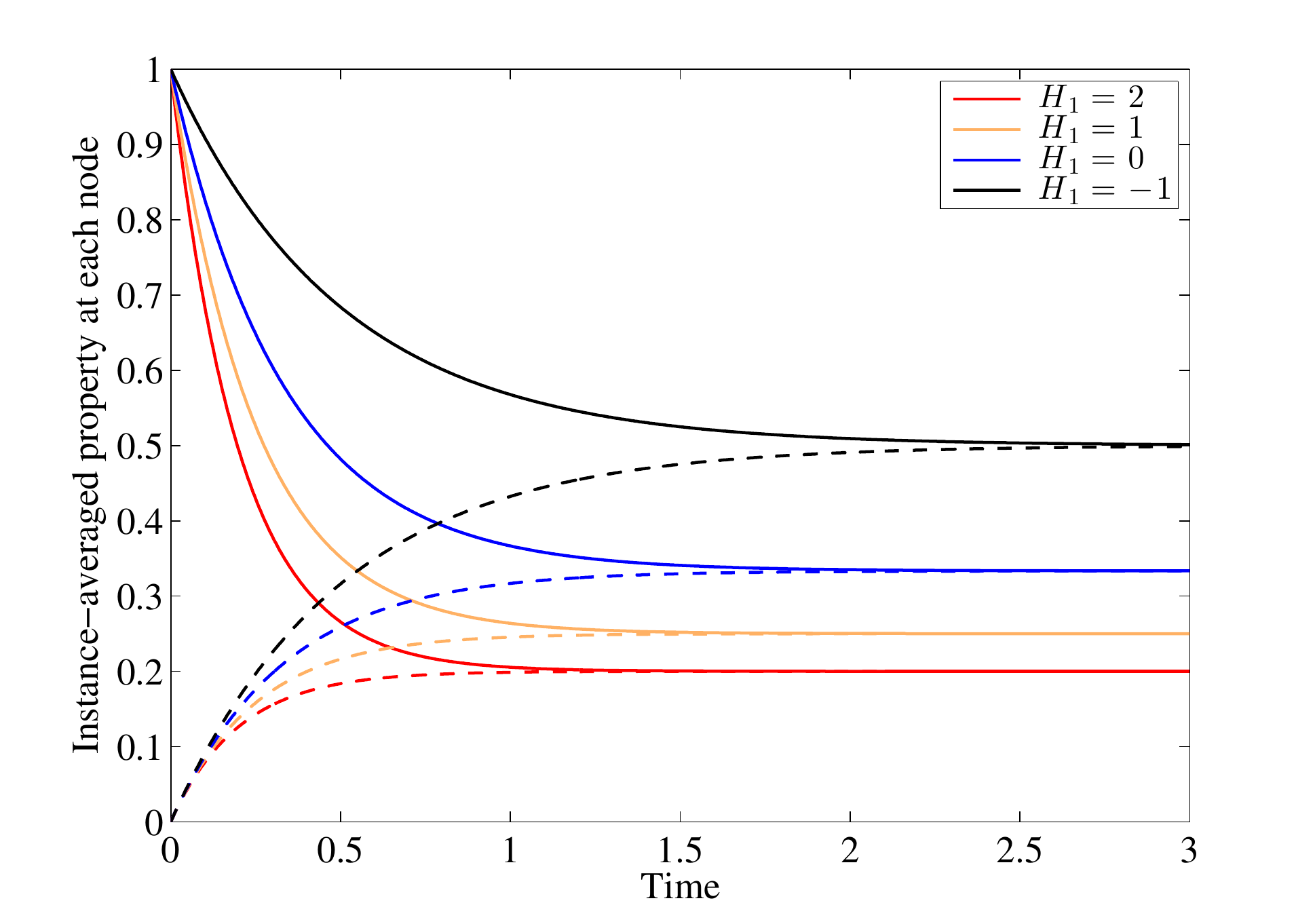}
		\label{fig:propCtrlConsens}
	}
	\caption{Resultant outputs for different $H_1$ for the two network update protocols by mimicking proportional control on quasi-mode 1. The solid lines indicate the time evolution of properties at nodes 1 and 3, while the dashed lines indicate the time evolution of properties at nodes 2 and 4.}
	\label{fig:propCtrl}
\end{figure}
We perform proportional control on quasi-mode 1 by replacing $H_1$ with a fixed constant. We illustrate the resultant outputs for some constants $H_1$ in Figures \ref{fig:propCtrlConserv} and \ref{fig:propCtrlConsens} for the conservative and non-conservative cases. %When $H_1=-2$, nodes 1 and 3 reach a steady average output of 1 with zero response time, while the other nodes reach their steady average output after three times the original characteristic time.
By decreasing $H_1$ towards $-3$, we can delay the onset of the steady state, but we also increase the influence of the inputs on the system by increasing the total system input. %and we illustrate this using $H_1=-2.95$. This influence affects the total steady-state property in a conservative system and the steady-state consensus value in a non-conservative system.
%Note that in Table \ref{table:propCtrlConsens}, the total system input influences the extent to which the final consensus value is increased due to the external action.

\subsubsection{Integral control}
\begin{figure} [!t]
	\centering
	\subfloat [Conservative protocol]{
		\includegraphics[width=0.5\linewidth]{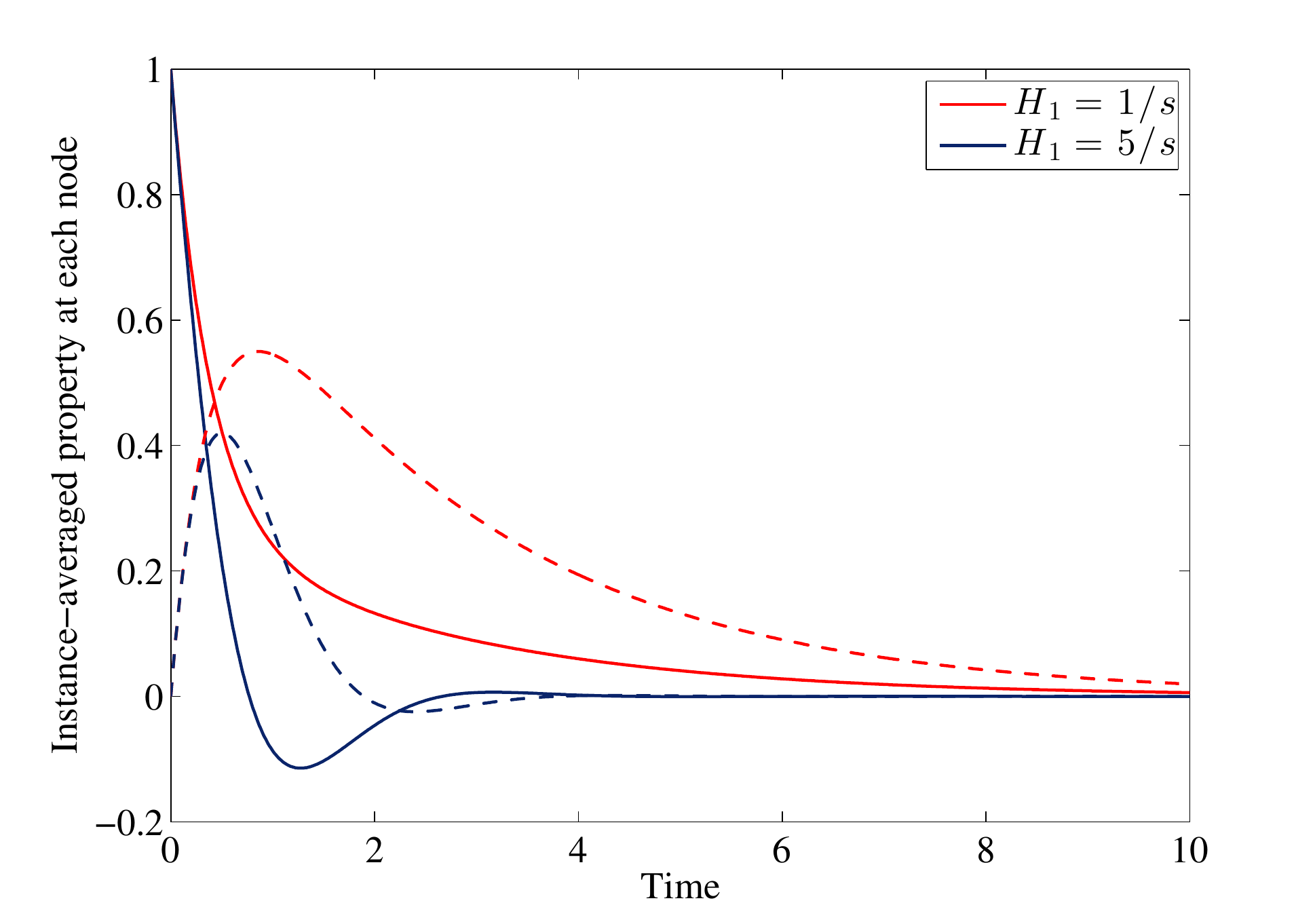}
		\label{fig:intCtrlConserv}
	}
	\subfloat [Non-conservative protocol]{
		\includegraphics[width=0.5\linewidth]{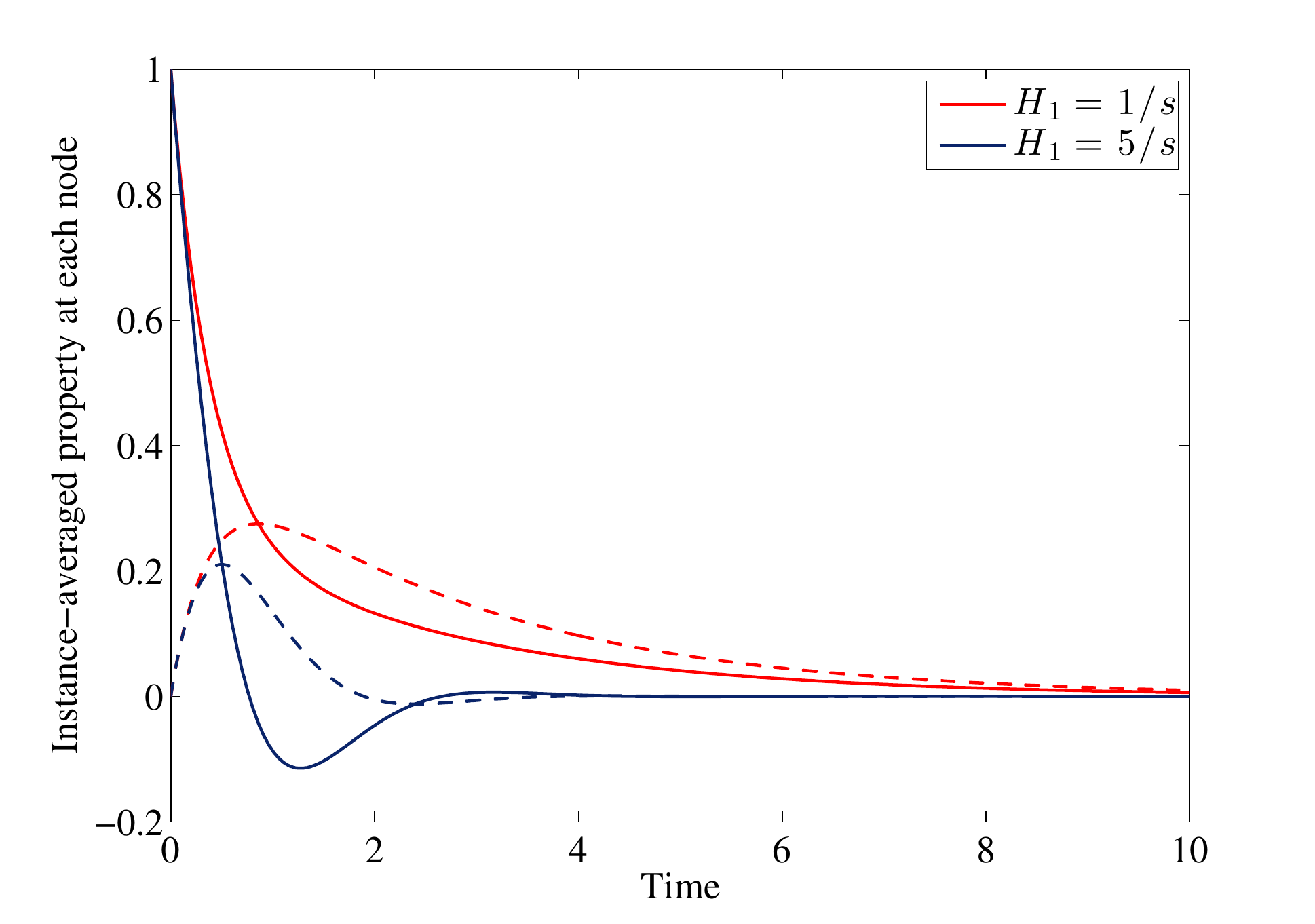}
		\label{fig:intCtrlConsens}
	}
	\caption{Resultant outputs for different $H_1$ for the two network update protocols by mimicking integral control on quasi-mode 1. The solid lines indicate the time evolution of properties at nodes 1 and 3, while the dashed lines indicate the time evolution of properties at nodes 2 and 4.}
	\label{fig:intCtrl}
\end{figure}
In standard control theory, proportional control is associated with a steady-state error relative to the input. This causes the outputs to be maintained at a non-zero level long after the impulses have ended. Integral control eliminates this error and further forces the steady outputs to zero at infinite time. We verify this possibly useful behavior in Figures \ref{fig:intCtrlConserv} and \ref{fig:intCtrlConsens} for the conservative and non-conservative cases by performing integral control on quasi-mode 1. The transient actions neither influence the total steady-state property in a conservative system nor the steady-state consensus value in a non-conservative system. However, the magnitude of the integral gain can influence the speed of onset of the steady state.%, as well as the total system input.

\subsubsection{Conclusions on external control}
We have demonstrated that by varying the input functions appropriately, both the transient and steady-state perturbations to the system can be modified. By careful design cognizant of the characteristic timescales of the system, the timescales of the perturbations can be varied. In addition, the perturbations can be designed with either permanent or temporary effects on the system depending on the objectives of the network designer.

\subsection{Network classification}
\begin{figure*} [!t]
	\centering
	\subfloat [Erd\"{o}s-R\'{e}nyi (random graph)]{
		\includegraphics[width=0.25\linewidth]{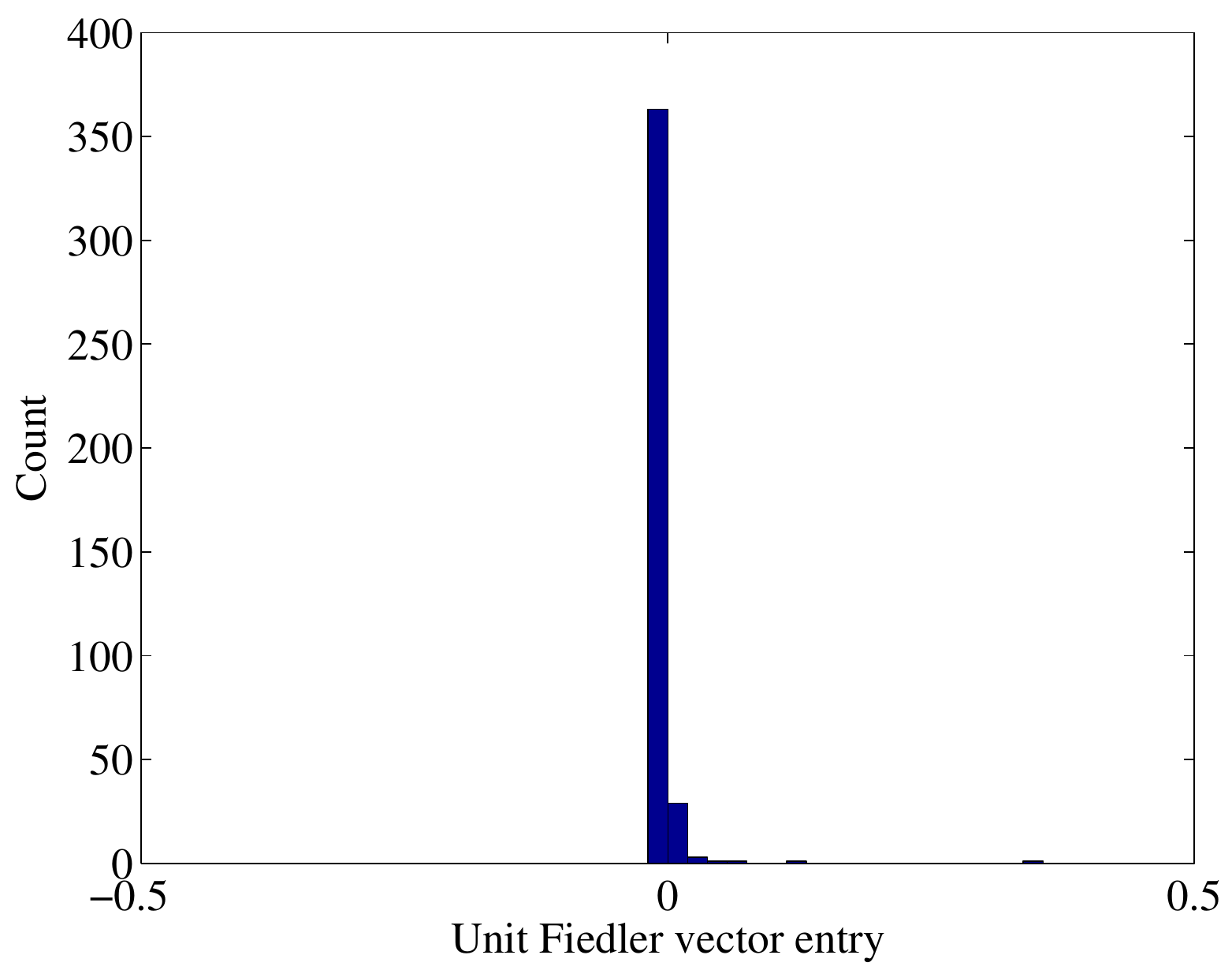}
	}
	\subfloat [Albert-Barab\'{a}si (scale-free)]{
		\includegraphics[width=0.25\linewidth]{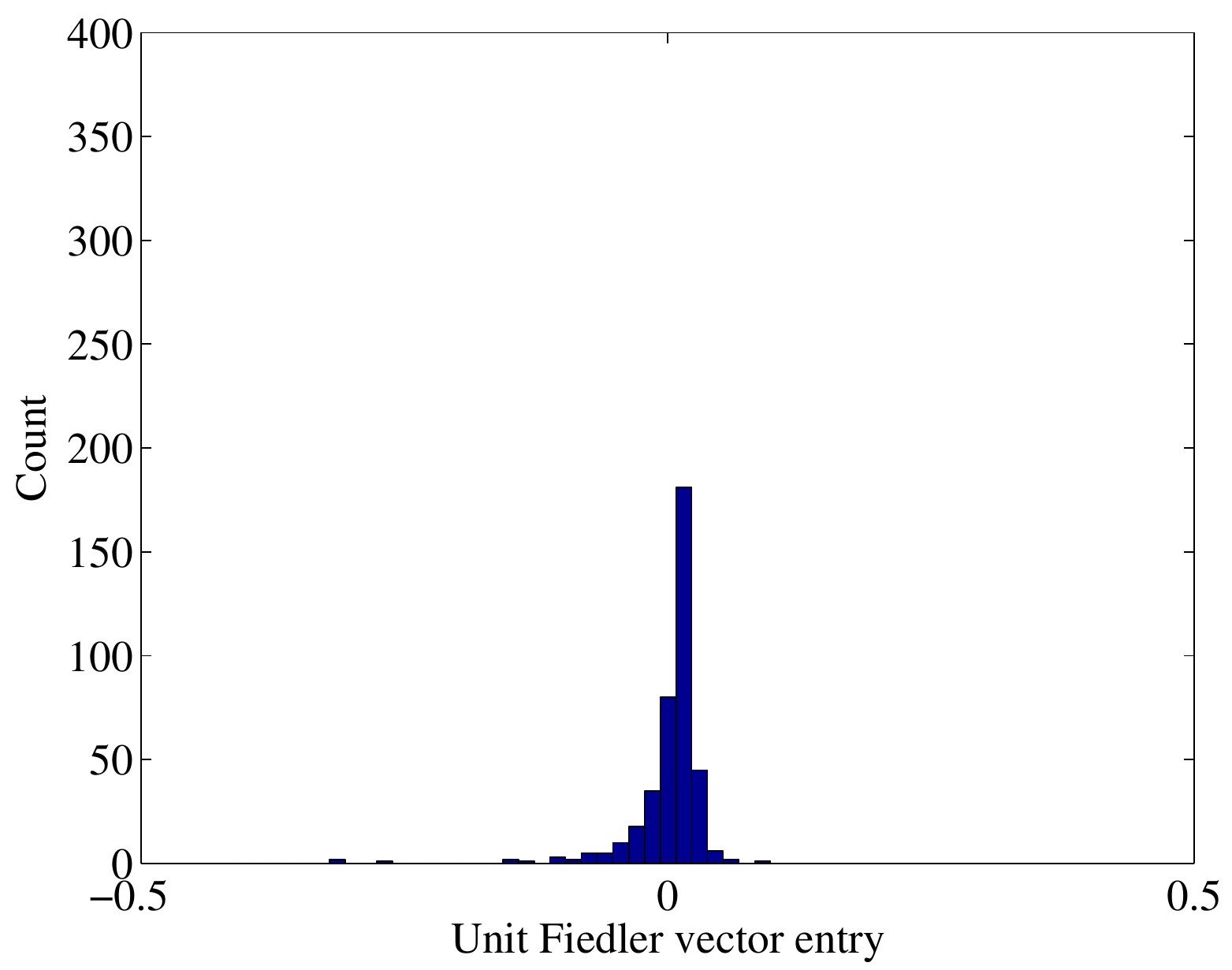}
	}
	\subfloat [Watts-Strogatz (small-world)]{
		\includegraphics[width=0.25\linewidth]{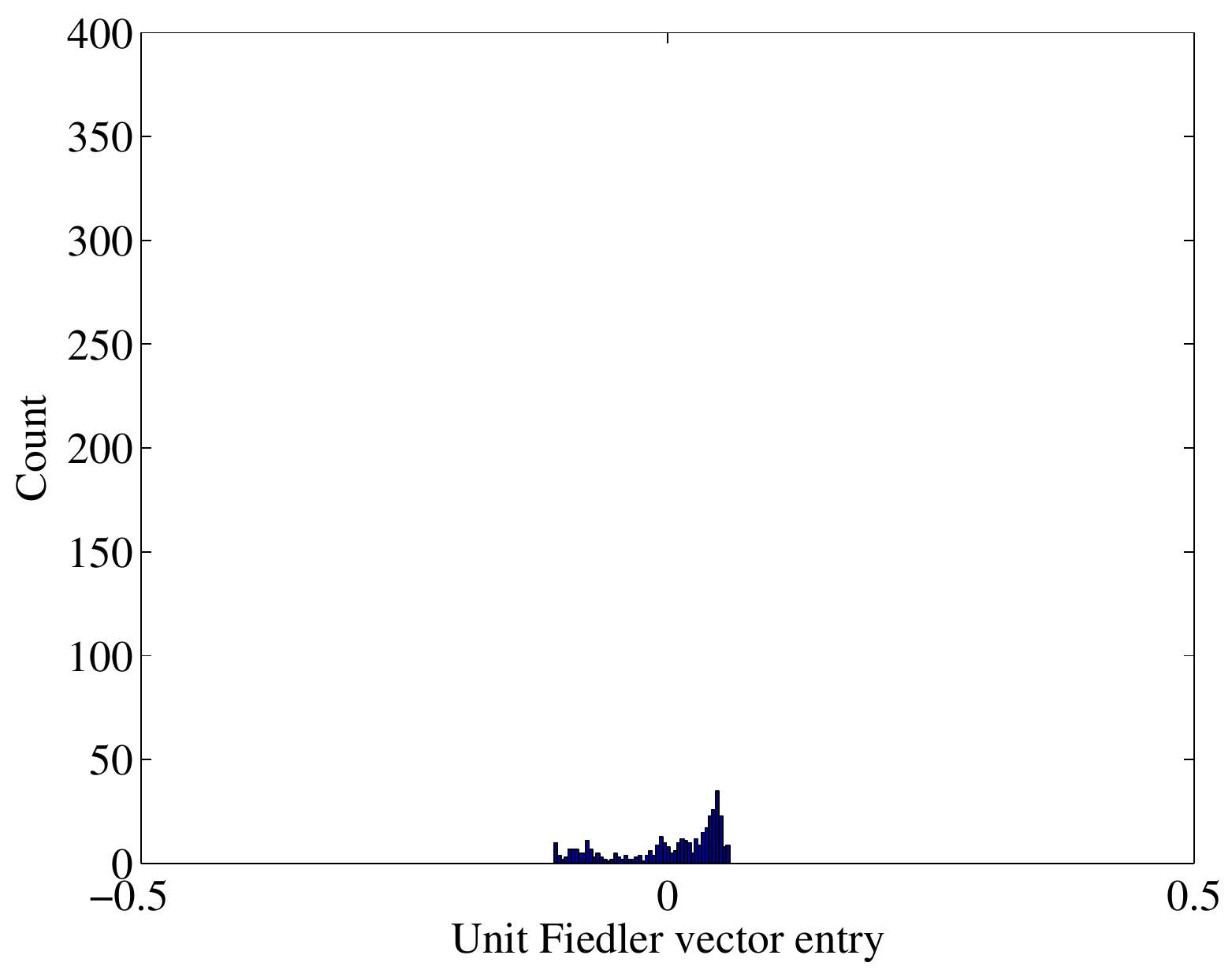}
	}
	\caption{Histograms of Fiedler vector entries for a single instance of various random networks with 400 nodes using the CONTEST toolbox~\cite{taylor2009contest}.}
	\label{fig:classNet}
\end{figure*}
Section \ref{sec:extCtrl} illustrated that both the external input and the structure of the quasi-modes in the system influence the resultant system outputs. The variance in the individual entries of the eigenvectors corresponding to the activated quasi-modes influence the effects of the external input on the individual nodes. For example, an external input on a single node in a network where the activated quasi-modes have low variances in their eigenvector entries tends to influence all other connected nodes to the same extent; conversely, the same input in a network with quasi-modes with large variances may have different effects on different nodes. 

We provide a simple example here to illustrate this point. It is generally accepted that the eigenvector associated with the non-zero eigenvalue of the Laplacian with the smallest absolute value, the Fiedler vector, provides a reasonable basis for spectral graph partitioning~\cite{fiedler1973algebraic}. We plot histograms of the entries of the Fiedler vector for different classes of networks in Fig. \ref{fig:classNet} to illustrate their varying distributions. We observe that the spread of the Fiedler vector entries increases with the average network clustering coefficient. Since small-world networks tend to have Fiedler vector entries with a larger spread, inputs to small-world networks can affect different nodes to different extents. The input design considerations for small-world networks have to ensure that specific nodes are targeted to achieve the desired control objective.

\section{Control by structure modification}
\label{sec:Internal}
Instead of controlling the network dynamics by means of exogenous excitation, we present in this section two methods to steer the dynamics by modifications in the network structure. 
\subsection{Network structure modification through design}
Instead of controlling the system with an exogenous input term, we can also modify the network structure directly to achieve control. Modifying the network structure may be preferable to introducing external excitations in systems where it is more convenient to modify the interactions between agents than to introduce or remove property from agents. %For example, if the property in question was to be a valuable commodity like a financial instrument, then modifying the inter-agent interactions could prove to be more cost-effective. 
Through this modification, we can shift both the eigenvalues and the eigenvector components, thereby changing the system dynamics as well as its steady-state distribution. In particular, modifying the network links causes changes in both the transformation basis between node space and quasi-node space, as well as the characteristic timescales of each system mode. 

By modifying the links more carefully, we can achieve modification of the characteristic timescales without reshaping the transformation basis. Recall that $\mathcal{Q} = A \Lambda A^{-1}$ if it is diagonalizable. Apart from decomposing $\mathcal{Q}$ into the matrices $A$, $\Lambda$ and $A^{-1}$, we can also reconstruct $\mathcal{Q}$ by multiplying $A$, $\Lambda$ and $A^{-1}$ together in the appropriate order. By holding $A$ constant and modifying the entries of $\Lambda$, we can achieve a modification in $\mathcal{Q}$ without affecting the composition of the system modes. The stability of the new eigenvalues of the modified $\mathcal{Q}$ will determine if the modified $\mathcal{Q}$ continues to represent a CTMC described by \eqref{eqn:markov2}. In this respect, as long as $q_\mathrm{s} = 0$ is preserved, the modification of $\mathcal{Q}$ will result in a CTMC where the columns or rows sum to zero for conservative and non-conservative networks, respectively.

\begin{Remark}
We note that this approach of network structure modification through design can be extended to more complex systems. In particular, consider the distributed control system~\cite{burbano2014distributed} 
\begin{equation}
\ddot{\bar{\mathcal{T}}} = \mathcal{Q}(\beta'_D \ddot{\bar{\mathcal{T}}} + \beta'_P \dot{\bar{\mathcal{T}}} + \beta'_I \bar{\mathcal{T}}) \, ,
\end{equation}
for some arbitrary gains $\beta'_D$, $\beta'_P$ and $\beta'_I$. Here, modifications to $\mathcal{Q}$ will also influence the dynamics of the system. However, the key governing parameters are no longer the eigenvalues $q_k$ of $\mathcal{Q}$, but instead the roots of the characteristic polynomial of the matrix $\left[ s^2(\mathbb{I} - \beta'_D \mathcal{Q}) - s(\beta'_P \mathcal{Q}) - \beta'_I \mathcal{Q} \right]$.
\end{Remark}

\begin{figure}[!t]
\centering
\begin{tikzpicture}[->,auto,node distance=1.3cm,>=latex',graph node/.style={circle,draw},every node/.style={scale=0.65}]
	
	\node [graph node] (1) {1};
	\node [graph node] (2) [above of=1, left of=1] {2};
	\node [graph node] (3) [above of=1, right of=1] {3};
	\node [graph node] (4) [below of=1, left of=1] {4};
	\node [graph node] (5) [below of=1, right of=1] {5};
	
	\path (1) edge [bend left] node [pos=0.8] {1} (2);
	\path (1) edge [bend left] node [pos=0.8] {1} (3);
	\path (1) edge [bend left] node [pos=0.8] {1} (4);
	\path (1) edge [bend left] node [pos=0.8] {1} (5);
	\path (5) edge [bend left] node [pos=0.2] {1} (1);
	\path (4) edge [bend left] node [pos=0.2] {1} (1);
	\path (3) edge [bend left] node [pos=0.2] {1} (1);
	\path (2) edge [bend left] node [pos=0.2] {1} (1);
	
\end{tikzpicture}
\caption{Star graph $S_5$ with symmetric links of unit weights.}
\label{fig:modifyNetwork}
\end{figure}
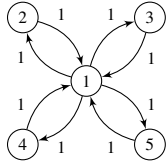
To illustrate how we can modify the eigenvalues of a network to achieve a shift in the system response, we consider a symmetric star network depicted in Fig. \ref{fig:modifyNetwork}. The transition rate matrix $\mathcal{Q}$ corresponding to this network is  
\begin{equation*}
\mathcal{Q} = 
\begin{bmatrix}
-4 & 1 & 1 & 1 & 1 \\
1 & -1 & 0 & 0 & 0 \\
1 & 0 & -1 & 0 & 0 \\
1 & 0 & 0 & -1 & 0 \\
1 & 0 & 0 & 0 & -1
\end{bmatrix}\, .
\end{equation*}
The eigenvalues of this matrix are 0 with multiplicity 1, $-5$ with multiplicity 1 and $-1$ with multiplicity 3. The corresponding eigenvectors are the steady-state eigenvector describing equilibration between all the nodes ($\left[\begin{smallmatrix}1 & 1 & 1 & 1 & 1\end{smallmatrix}\right]$), an eigenvector describing diffusion between the peripheral nodes and the central node ($\left[\begin{smallmatrix}-4 & 1 & 1 & 1 & 1\end{smallmatrix}\right]$), and eigenvectors describing diffusion among the peripheral nodes ($\left[\begin{smallmatrix}0 & 1 & 1 & 1 & -3\end{smallmatrix}\right]$, $\left[\begin{smallmatrix}0 & 0.36 & 1 & -1.36 & 0\end{smallmatrix}\right]$ and $\left[\begin{smallmatrix}0 & -1.36 & 1 & 0.36 & 0\end{smallmatrix}\right]$). Suppose we want to slow down diffusion between the peripheral nodes and the central node. We can do so by decreasing the effective update rate of this mode, which causes the relative importance of diffusion among the peripheral nodes to increase and results in extra links facilitating this diffusion. For instance, by changing the effective update rate of the peripheral-central mode from 5 to 4.5, we obtain the following transition rate matrix
\begin{equation}
\mathcal{Q} = 
\begin{bmatrix}
-3.6 & 0.9 & 0.9 & 0.9 & 0.9 \\
0.9 & -0.975 & 0.025 & 0.025 & 0.025 \\
0.9 & 0.025 & -0.975 & 0.025 & 0.025 \\
0.9 & 0.025 & 0.025 & -0.975 & 0.025 \\
0.9 & 0.025 & 0.025 & 0.025 & -0.975
\end{bmatrix}\, .
\label{eqn:modifyQ2}
\end{equation}
In the resultant architecture, the links between the central node and the peripheral nodes have weights of 0.9 instead of 1, and the links between the peripheral nodes have weights of 0.025 instead of 0. The corresponding network is a superposition of the two graphs depicted in Fig. \ref{fig:modifyNetwork2}.

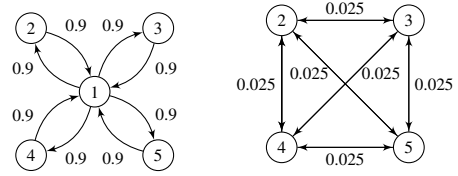
\begin{figure}[!t]
\centering
\begin{tikzpicture}[->,auto,node distance=1.3cm,>=latex',graph node/.style={circle,draw},every node/.style={scale=0.65}]
	
	\node [graph node] (1) {1};
	\node [graph node] (2) [above of=1, left of=1] {2};
	\node [graph node] (3) [above of=1, right of=1] {3};
	\node [graph node] (4) [below of=1, left of=1] {4};
	\node [graph node] (5) [below of=1, right of=1] {5};
	
	\path (1) edge [bend left] node [pos=0.8] {0.9} (2);
	\path (1) edge [bend left] node [pos=0.8] {0.9} (3);
	\path (1) edge [bend left] node [pos=0.8] {0.9} (4);
	\path (1) edge [bend left] node [pos=0.8] {0.9} (5);
	\path (5) edge [bend left] node [pos=0.2] {0.9} (1);
	\path (4) edge [bend left] node [pos=0.2] {0.9} (1);
	\path (3) edge [bend left] node [pos=0.2] {0.9} (1);
	\path (2) edge [bend left] node [pos=0.2] {0.9} (1);

\end{tikzpicture}
\hspace{4mm}
\begin{tikzpicture}[->,auto,node distance=2.6cm,>=latex',graph node/.style={circle,draw},every node/.style={scale=0.65}]
	
	\node [graph node] (1) {2};
	\node [graph node] (2) [right of=1] {3};
	\node [graph node] (3) [below of=2] {5};
	\node [graph node] (4) [left of=3] {4};
	
	\path (1) edge node {0.025} (2);
	\path (2) edge node {0.025} (3);
	\path (3) edge node {0.025} (4);
	\path (4) edge node {0.025} (1);
	\path (1) edge node [left, pos=0.4] {0.025} (3);
	\path (1) edge (4);
	\path (2) edge node [right, pos=0.4] {0.025} (4);
	\path (3) edge (1);
	\path (4) edge (1);
	\path (4) edge (2);
	\path (1) edge (4);
	\path (4) edge (3);
	\path (3) edge (2);
	\path (2) edge (1);
		
\end{tikzpicture}
\caption{The superposition of these two graphs yields the network structure corresponding to the modified transition rate matrix given in \eqref{eqn:modifyQ2}.} 
\label{fig:modifyNetwork2}
\end{figure}

%%%%%%%%%%%%%%%%%%%%%%%%%%%%%%%%%%%%%%%%
\subsection{Network structure modification through adaptive control}

In sufficiently large systems, it can be computationally intensive to perform the eigendecomposition of various $\mathcal{Q}$ matrices in order to determine the network structure that gives the most desirable system response. There are two possible approaches to handle this issue. Firstly, one can choose from a set of matrices that deviate from the original $\mathcal{Q}$ by small perturbations $\delta\mathcal{Q}$. The eigenvalue and eigenvector perturbations can be derived as a function of $\delta\mathcal{Q}$ for sufficiently small $\delta\mathcal{Q}$. However, this limits us to small modifications in the network structure that do not span the feasible action space. 

It is possible to examine larger and/or more modifications without going through the computational complexity by performing adaptive control driven by reinforcement learning. In this approach, we formulate the problem as an MDP. The state space $\mathcal{X}$ of the MDP is equivalent to the node space of the network, while the action space $\mathcal{W}$ of the MDP is equivalent to the set of possible $\mathcal{Q}$ that the network can adopt. More formally, $\mathcal{X} = \mathcal{V}$, while $\mathcal{W} = \{\mathcal{Q}_1, \ldots, \mathcal{Q}_{w^*}\}$ for $w^*$ different actions. In essence, the MDP is a direct extension of the CTMC upon which the governing equation of the network diffusion is based.\footnote{A MDP with a single action is equivalent to a CTMC.} By selecting a reward function $R(X,W)$ that describes the desirability of action $W$ given that the system is in state $X$, and by scheduling decisions to be made immediately after a change in the state of the system\footnote{This means that the variation of the active action $W(t)$ with time is a piecewise constant function.}, the MDP is completely defined. We can determine the conditions for optimality in the MDP by searching the policy that yields the optimal expected reward integrated over time
\begin{equation}
\mathbb{E}_{X(0),W(0),\text{ sample paths}}\left[\int_0^\infty \gamma^t R(X(t),W(t)) dt \right] \, ,
\label{eqn:optimalReward}
\end{equation}
where $\gamma \in (0,1)$ denotes a discount rate. By considering (\ref{eqn:markov2}) and (\ref{eqn:optimalReward}) simultaneously, we obtain the Hamilton-Jacobi-Bellman equation for the system. In optimal control theory, this equation is solved to obtain the optimal policy for the system. For large state and/or action spaces, this solution can be difficult to obtain. For a more computationally manageable approach, we turn to reinforcement learning to determine a suboptimal solution with sufficiently fast convergence. For each state-action pair $(X,W)$, we store a quality value $V^\mathcal{Q}(X,W)$, and update $V^\mathcal{Q}(X,W)$ based on the Q-learning method~\cite{barto1998reinforcement}
\begin{multline}
V^\mathcal{Q}_{k+1}(X,W) = (1 - \mu) V^\mathcal{Q}_{k} (X, W) \\
+ \mu \left[ \max_{W'} R(X_\text{next},W') + \gamma \max_{W'} V^\mathcal{Q}_{k}(X_\text{next},W') \right] \, ,
\end{multline}
where $k$ is the index indicating the $k$-th change of the system state, $\mu \in (0,1]$ is the learning rate of the algorithm, and $X_\text{next}$ is the next state of the system. %which is determined by generating random hitting times for all other states based on the mean hitting times defined by $W$, and then selecting the state with the smallest hitting time. 
Whenever the system transitions to a new state, a new action has to be selected. We set up the system such that with probability $\epsilon$, the system selects the action with the highest quality given its current state, while with probability $1-\epsilon$, the system selects a random action. This is known as the $\epsilon$-greedy policy where $\epsilon$ is the exploitation probability of the system. By selecting appropriate values of $\mu$, $\gamma$ and $\epsilon$, as well as an appropriate reward function $R(X,W)$, we can design the learning process to achieve predefined system objectives at a suitable convergence rate.

%\subsubsection{Case study, conservative update rule}

To illustrate the use of an MDP to select a desirable $\mathcal{Q}$ in a large network, we set up a state space $\mathcal{X}$ containing 20 agents, and we set up the action space $\mathcal{W}$ by generating 400 $\mathcal{Q}$ matrices corresponding to 400 complete graphs whose links each have weights sampled from a uniform distribution on $[0,1]$, to simulate 400 different random network configurations. These matrices obey the conservative relation \eqref{eqn:CTMC1Conservative}, which obeys the update rule \eqref{eqn:P1}. In addition, we select 2 arbitrary nodes (12 and 14) out of the 20 to serve as target nodes for our system. We let the state transitions of the MDP be governed by a grand system transition matrix $\mathcal{Q}_g$. Whenever a state transition occurs, we assign a reward of 5 whenever the state coincides with 1 of these 2 target nodes, and a reward of 0 otherwise, in order to encourage the system to visit these states more often. The MDP then selects an action, following which $\mathcal{Q}_g$ is updated based on this selected action. For example, if the updated state of the system is node 5 and action 20 is selected, then the 5th column of $\mathcal{Q}_g$ is replaced by the 5th column of $\mathcal{Q}_{20}$ corresponding to the 20th action.

\begin{figure}[!t]
\centering
\includegraphics[width=0.8\linewidth]{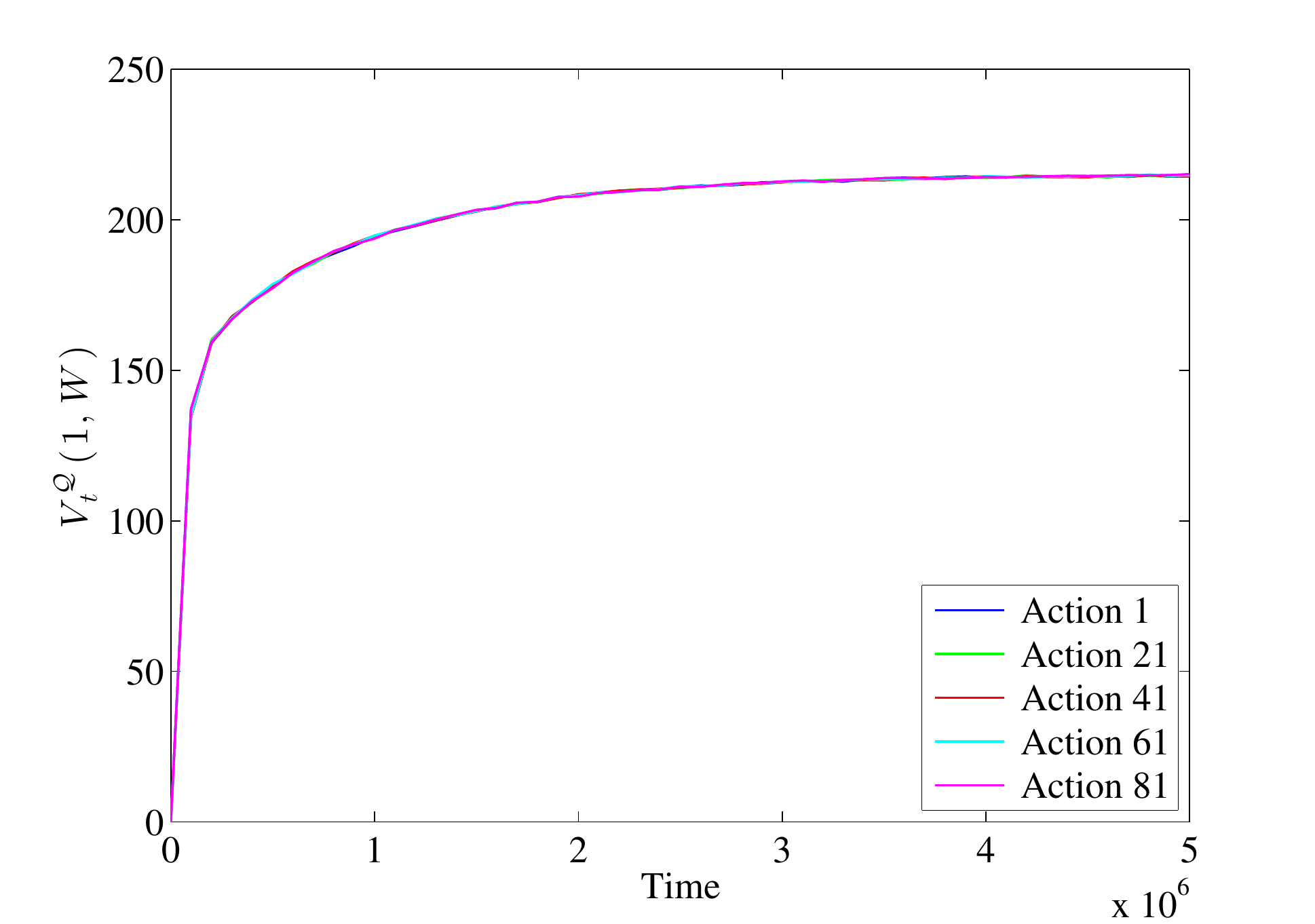}
\caption{Time evolution of quality $V^\mathcal{Q}(1,W)$ of node 1 for 5 different actions $W$ for a single trial.}
\label{fig:MDP1}
\end{figure}
\begin{figure}[!t]
\centering
\includegraphics[width=0.8\linewidth]{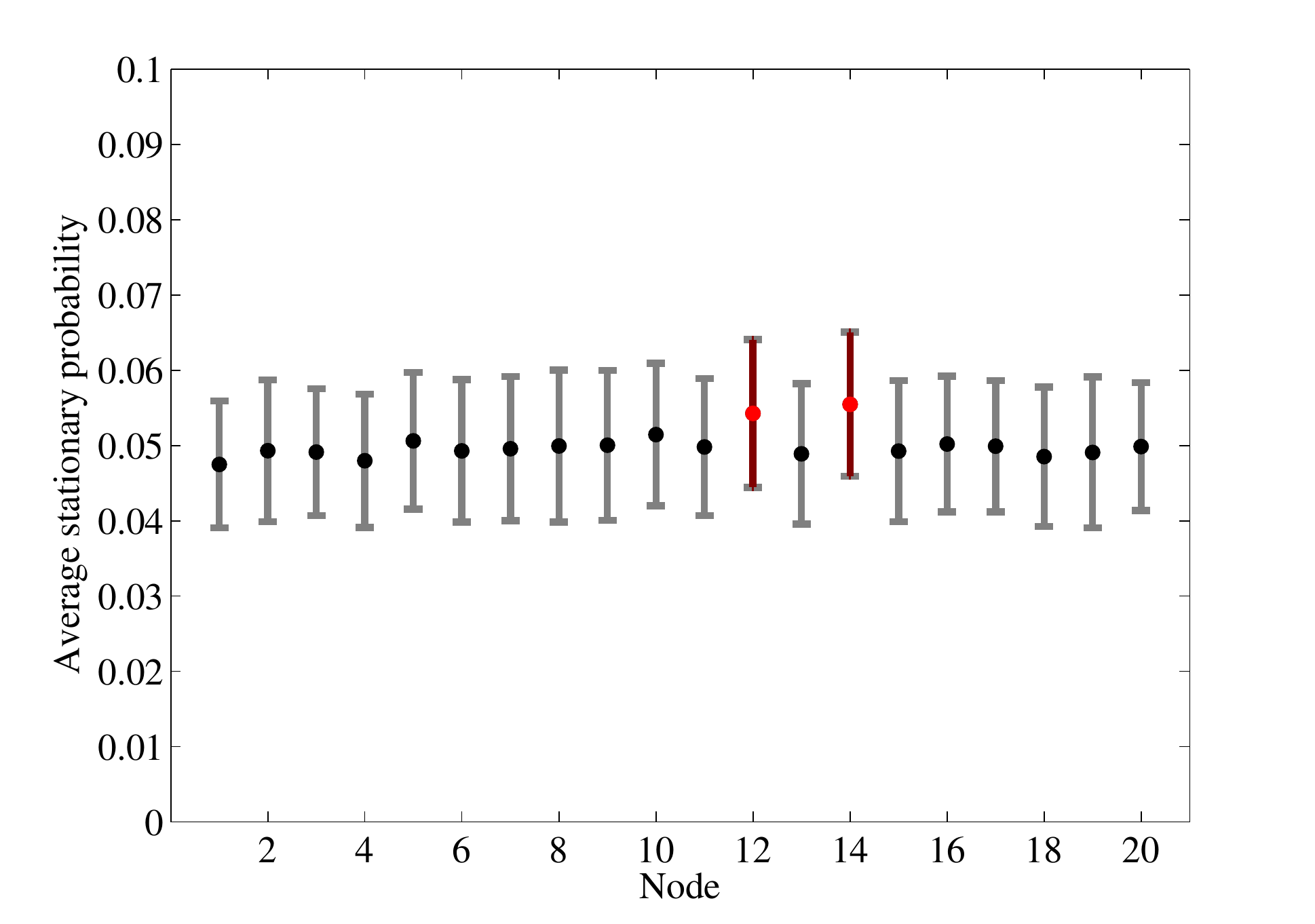}
\caption{Stationary probability distribution of the simulated network averaged over 100 trials. The nodes marked red are the two randomly selected nodes (12 and 14) that were assigned a reward on every visit by each MDP.}
\label{fig:MDP2}
\end{figure}
For this work, we ran 100 independent numerical simulations describing the time evolution of 100 MDPs with the same state and action spaces, but with different initial states and $\mathcal{Q}_g$. Using the reinforcement learning parameters $\mu = 0.2$, $\epsilon = 0.4$ and $\gamma = 0.995$, we run each simulation for $5\text{\sc{E}}6$ time-steps, following which we find the steady-state eigenvector (normalized to sum to 1), or stationary distribution, $v_0$ for each last known $\mathcal{Q}_g$ such that $\mathcal{Q}_g v_0=0$ and $\sum_i v_0(i) = 1$. In Fig. \ref{fig:MDP1}, we demonstrate using the time evolution of $V^\mathcal{Q}(1,W)$ for 5 different actions $W$ that the learning curve has reached convergence. In Fig. \ref{fig:MDP2}, we plot the last known stationary distribution $v_0(i)$ over all the states $i$ averaged over the 100 trials, and demonstrate that the MDP is, on average, able to modify the stationary distribution of the system after sufficient time has elapsed by careful design of the reward function. 

In this case, we guided the system to favor nodes 12 and 14 by modifying the network structure, analogous to how we modified the system steady state in Section \ref{sec:extCtrl} through the injection of appropriate external inputs. This MDP methodology can either be used online to allow the network to respond to time-varying objectives, or offline to cycle through various network possibilities with reduced computational complexity.

%\subsubsection{Case study, consensus update rule}
%
%\begin{figure}[!t]
%\centering
%\includegraphics[width=0.9\linewidth]{Figures/consensusMDP4}
%\caption{Histogram of final consensus value for the simulated network over 100 trials using a reinforcement learning algorithm.}
%\label{fig:MDP3}
%\end{figure}
%
%\begin{figure}[!t]
%\centering
%\includegraphics[width=0.9\linewidth]{Figures/consensusRand4}
%\caption{Histogram of final consensus value for the simulated network over 100 trials using a random but constant policy for each trial.}
%\label{fig:MDP4}
%\end{figure}
%
%We also adapted the case study above for the consensus update rule by having the actions obey the relation \eqref{eqn:CTMC1Consensus} instead of \eqref{eqn:CTMC1Conservative}. In addition, we set up a parallel simulation simulating the actual property at each node, setting the initial property to be 20 at nodes 12 and 14, and 1 at all other nodes. Figures \ref{fig:MDP3} and \ref{fig:MDP4} illustrate the final consensus values over 100 independent numerical simulations where a reinforcement learning algorithm and a random policy were used respectively. Note that in the case of the reinforcement learning algorithm, the system is able to obtain higher consensus values more frequently, but with the trade-off that it also obtains lower consensus values more frequently.
%\input{applications}
\section{Conclusion}
\label{sec:Conclusion}

In this work, we proposed a probabilistic framework that describes diffusion in multi-agent networks and allows for a variety of inter-agent update rules. We focused on two protocols where the total network quantity is conserved (conservative protocol) or variable (non-conservative protocol). By including the possibility of asymmetric updates, non-unitary links and switching topologies, we examined the stability and convergence characteristics of generic networks, and described their transient and steady-state dynamics in detail. In addition, we modeled external influences on networks by considering an inhomogeneous form of the governing equation. As such, the framework has the potential to capture the presence of stubborn agents and the process of dynamic learning. Furthermore, we demonstrated the ability to achieve control of diffusion in these networks using either external excitation or modification of the network structure. For the former, we showed that the form of the external excitation and the network structure are crucial to the characteristics of the controlled output, and for the latter, we identified an algorithm involving reinforcement learning for an MDP to promote quick achievement of the control objective. Through these techniques, we enable the possibility of actuating individual nodes and steering the system steady state towards some desired control objective, thereby allowing for the micromanagement of diffusion in networks. 

Our future work to extend the current toolbox will include more direct applications to large random networks. In addition, we intend to include time-inhomogeneous and stochastic governing equations, and to extend the framework to continuous node spaces. Another important extension of this work is to develop more insight in the higher order moments of the property.

\bibliographystyle{IEEEtran}
\bibliography{references}

\end{document}